\documentclass[conference]{IEEEtran}
\pdfoutput=1
\IEEEoverridecommandlockouts
\usepackage{cite}
\usepackage{amsmath}
\usepackage{amssymb,amsfonts}
\usepackage[noend]{algpseudocode}
\usepackage{algorithm}

\usepackage{graphicx}
\usepackage{textcomp}
\usepackage{xcolor}
\usepackage{subfigure}
\usepackage{multirow}
\usepackage{soul}
\usepackage{amsthm}
\usepackage{tikz}
\usetikzlibrary{decorations.pathreplacing,calc}
\newcommand{\tikzmark}[1]{\tikz[overlay,remember picture] \node (#1) {};}

\newcommand*{\SpaceReservedForComments}{2.5cm}%
\newcommand*{\HorizontalOffset}{-0.5em}%
\newcommand*{\VerticalOffset}{0.7ex}%
\newcommand*{\AddNote}[4][]{%
    \begin{tikzpicture}[overlay, remember picture]
        \draw [decoration={brace,mirror, amplitude=0.5em},decorate,thick, red, #1]
            ($(#3)+(\HorizontalOffset + 0.80*\linewidth,+\VerticalOffset)$) --  ($(#2)+(\HorizontalOffset + 0.80*\linewidth,\VerticalOffset)$)
            node [align=left, text width=\SpaceReservedForComments-1.0em, pos=0.5, anchor=west] {#4};
    \end{tikzpicture}
}%

\makeatletter% Add a \tkizmark for each line so we can reference it later
    \algrenewcommand\alglinenumber[1]{\tikzmark{\arabic{ALG@line}}\tiny#1:}
\makeatother

\def\BibTeX{{\rm B\kern-.05em{\sc i\kern-.025em b}\kern-.08em
    T\kern-.1667em\lower.7ex\hbox{E}\kern-.125emX}}
\begin{document}

\newtheorem{definition}{Definition}
\newtheorem{example}{Example}
\newtheorem{theorem}{Theorem}
\newtheorem{problem}{Problem}
%\newtheorem{proposition}{Proposition}
% \newtheorem*{Proof}{Proof}

%\theoremstyle{nonumberplain}
%\newtheorem{proof}{Proof}

%\title{A Cluster-based Restreaming Algorithm for Scalable Web Graph Partitioning}
\title{Clustering-based Partitioning for Large Web Graphs}

\author{\IEEEauthorblockN{\IEEEauthorrefmark{1}Deyu Kong, \IEEEauthorrefmark{2}Xike Xie and \IEEEauthorrefmark{3}Zhuoxu Zhang}

\IEEEauthorblockA{\IEEEauthorrefmark{1}\IEEEauthorrefmark{2}\IEEEauthorrefmark{3}University of Science and Technology of China}
% \IEEEauthorblockA{University of Science and Technology of China}

% \{cavegf, zzx371479326\}@mail.ustc.edu.cn}

\IEEEauthorblockA{\{\IEEEauthorrefmark{1}cavegf,\IEEEauthorrefmark{3}zzx371479326\}@mail.ustc.edu.cn, \IEEEauthorrefmark{2}xkxie@ustc.edu.cn}}
% \IEEEauthorblockA{\IEEEauthorrefmark{3}xkxie@ustc.edu.cn}}

\maketitle

\begin{abstract}
    Graph partitioning plays a vital role in distributed large-scale web graph analytics, such as pagerank and label propagation. The quality and scalability of partitioning strategy have a strong impact on such communication- and computation-intensive applications, since it drives the communication cost and the workload balance among distributed computing nodes. 
    Recently, the streaming model shows promise in optimizing graph partitioning.  
    However, existing streaming partitioning strategies either lack of adequate quality or fall short in scaling with a large number of partitions.

    In this work, we explore the property of web graph clustering and propose a novel restreaming algorithm for vertex-cut partitioning. We investigate a series of techniques, which are pipelined as three steps, streaming clustering, cluster partitioning, and partition transformation.
    More, these techniques can be adapted to a parallel mechanism for further acceleration of partitioning.
    Experiments on real datasets and real systems show that our algorithm outperforms state-of-the-art vertex-cut partitioning methods in large-scale web graph processing. Surprisingly, the runtime cost of our method can be an order of magnitude lower than that of one-pass streaming partitioning algorithms, when the number of partitions is large.
\end{abstract} 

%This paper aims to address such challenges. 


\section{Introduction}
\label{sec:introduction}

The scale of graphs grows with an unprecedented rapid pace, including web graphs, social graphs, biological networks, and so on.
Big graphs are often measured in terabytes or petabytes, with billions or trillions of nodes and edges.
To cope with the big graph challenge, many distributed graph system are developed, such as Pregel\cite{malewicz2010pregel}, PowerGraph \cite{gonzalez2012powergraph}, GraphX~\cite{graphx}, GraphLab \cite{low2012distributed}, and PowerLyra \cite{chen2019powerlyra}.
In these systems, a big graph is partitioned into a predefined number of subgraphs, which are stored in distributed nodes.
Each node of the distributed graph system operates on its subgraph in parallel, and different nodes are communicated and synchronized with message-passing~\cite{li2021gpugraphx}.
%Therefore, the efficacy, efficiency, and scalability of graph partitioning algorithms are found to be imperative ingredients for bulk synchronous iterative processing.
Therefore, the quality, efficiency, and scalability of graph partitioning algorithms are found to be imperative ingredients for bulk synchronous iterative processing in distributed systems. Because it affects the workload balancing and communication overheads, and thus has a direct effect on on large-scale graph system performance.

There are two mainstream graph partitioning strategies, {\it edge-cut} \cite{tsourakakis2014fennel,andreev2006balanced,krauthgamer2009partitioning,karypis1996parallel,slota2017pulp,restream} and {\it vertex-cut} \cite{petroni2015hdrf,zhang2017graph,xie2014distributed,margo2015scalable} partitioning, both of which are to optimize objectives of load-balancing and min-cut (for either edges or vertices), so that the overall performance of distributed graph systems can be improved.
The vertex-cut partitioning strategy evenly assigns graph edges to distributed machines in order to minimize the number of times that vertices are cut.
Theoretically and empirically, vertex-cut partitioning is proved to be significantly more effective than its counterpart for web graph processing \cite{nature, gonzalez2012powergraph}, because most real graphs follow power law distributions \cite{donato2004large}.
%Edge-cut partitioning would create many replicas for high-degree vertices degrading the communication efficiency and system balancing.

%are to trade off the total communication cost between computing nodes and balanced workloads on each machine.
%In practice, vertex-cut partitioning shows its superior partitioning property on web graphs compared to edge-cut partitioning \cite{gonzalez2012powergraph}, since many real-word web graphs have power-law degree distribution \cite{donato2004large}.

%Optimally partitioning large graphs is impracticable due to its NP-hardness \cite{feige2008improved}. More over, the goal of graph partitioning is to trade off the total communication cost between computing nodes and balanced workloads on each machine.

Despite many works done, the problem of effective graph partitioning on practical distributed graph system is still open.

The problem of graph partitioning has been widely studied in the past decade.
For vertex-cut partitioning, there are two categories, a) \emph{offline distributed algorithms} that load the complete graph into memory \cite{zhang2017graph,margo2015scalable,karypis1996parallel}, and b) \emph{online streaming algorithms} that ingest edges as streams and perform on-the-fly partitioning based on partial knowledge of the graph \cite{gonzalez2012powergraph,petroni2015hdrf,patwary2019window,hua2019quasi,xie2014distributed}.
Offline algorithms do not scale well for distributed graph systems, with the tremendous increase of data volumes. For example, METIS \cite{karypis1996parallel} requires more than $8.5$ hours to partition a graph with about $1.5$ billion edges to only $2$ partitions \cite{tsourakakis2014fennel}.
Online streaming algorithms consist of hashing-based methods (e.g. DBH \cite{xie2014distributed}, Hashing \cite{gonzalez2012powergraph}) and heuristic-based methods (e.g. Greedy \cite{gonzalez2012powergraph}, HDRF \cite{petroni2015hdrf}). The characteristics of vertex-cut streaming algorithms are summarized in Table~\ref{tab:characteristics}.
\vspace{-7pt}
\begin{table}[h]
\centering
    \caption{\centering \hspace{-2pt} Vertex-cut \hspace{-1pt} Streaming \hspace{-1pt}Partitioning\hspace{-1pt} Algorithms}
    \label{tab:characteristics}
    \begin{tabular}{c|c|c}
    % \toprule
    \hline
        \multirow{1}{*}{\textbf{Algorithm}} & \multirow{1}{*}{\textbf{Time Cost}} & \multirow{1}{*}{\textbf{Quality}} \\
        \hline Hashing \cite{gonzalez2012powergraph} &
        Low & Low \\
        DBH \cite{xie2014distributed} & Low & Low \\
        Mint \cite{hua2019quasi} & Medium & Medium \\
        Greedy \cite{gonzalez2012powergraph} & High & High \\
        HDRF \cite{petroni2015hdrf} & High & High \\
        CLUGP & Low & High \\
    % \bottomrule
    \hline
    \end{tabular}
\end{table}
% From Table~\ref{tab:characteristics}, it can be seen that heuristic-based methods achieve better partitioning quality than hashing-based methods, and perform better in bulk synchronous processing systems~\cite{abbas2018streaming}. {\it However, heuristic-based methods are hard to be parallelized and fall short in scaling with a larger number of partitions, because global storage of previous assignments has to be shared among workers \cite{pacaci2019experimental}. Hashing-based algorithms and Mint perform well in scalability but are inferior on partition quality.}

From Table~\ref{tab:characteristics}, it can be seen that heuristic-based methods achieve better partitioning quality than hashing-based methods, and perform better in bulk synchronous processing systems~\cite{abbas2018streaming}. However, heuristic-based methods are time-consuming, because a global status table needs to be locked each time a partition decision of an edge is made. Hashing-based methods and Mint perform faster than heuristic-based methods but are inferior in partition quality.

% Two shortcomings for heuristic-based methods, lacking of structural information, global state limitations in scalability.

%Two shortcomings for heuristic-based methods, poor performance on large number of partitions, computation bottleneck caused by global status table.

% the maintenance of global data structure incurs considerable computation and communication overheads.

To this end, we study the problem of vertex-cut partitioning for large-scale web graphs to propose a new versatile partitioning architecture. We tackle the performance and quality challenge by exploring the connections between graph clustering and partitioning \cite{girvan2002community, reichardt2007partitioning, agarwal2008modularity, yang2017hypergraph}. Our vision is to explore clustering for enhancing the partitioning quality, employ streaming techniques for improving the efficiency, and break the ties of global structures for boosting system performance.

Nevertheless, a series of technical challenges arise in confronting clustering-based vertex-cut partitioning.
First, existing streaming clustering techniques only work for edge-cut partitioning, so that a high-degree vertex can hardly be accurately identified with partial degree information. Once such vertices are falsely identified for cutting, many replicas would be generated deteriorating system balance and communication efficiency. More, it is infeasible for correcting the false cutting with low-cost subsequent compensation, since it takes much communication overhead for high-degree vertex retrieving and reshuffling.
Second, existing partitioning methods (e.g., HDRF \cite{petroni2015hdrf}) are highly dependent on the global structure of vertex degrees or partial degrees, hindering its extensibility to large-scale graph streaming scenarios.
The corresponding maintenance overhead becomes no more negligible, and even dominates the total time of graph application (e.g., pagerank) running on large partitions.

%In this paper, we present a cluster based offline streaming algorithm (CLUSP) for making partitioning decisions based on clusters not only on edges. CLUSP consists of three stps. First, we will apply the streaming graph cluster algorithm to generate fine-grained clusters. Second, we propose a game theory approach to assign these clusters into different partitions. At last, we will raise a heuristic method to convert the cluster partitioning result into edge partitioning result. Our main contributions are summarized as follows:

In our work, we present a \underline{CLU}stering-based restreaming \underline{G}raph \underline{P}artitioning (CLUGP) architecture for vertex-cut partitioning over large-scale web graphs. Our algorithm follows a novel three-pass restreaming framework, which is pipelined as three steps, streaming clustering, cluster partitioning, and partition transformation.
The streaming clustering step exploits the connection between clustering and vertex-cut partitioning for generating fine-grained clusters and reducing vertex replicas.
The cluster partitioning step applies game theories for mapping generated clusters into specific partitions and further refines clustered results.
Then, the partition transformation step transforms the cluster-based partitioning results into vertex-cut partitioning results.

%proposes a heuristic method with guaranteed error bounds for transforming the cluster-based partitioning results into vertex-cut partitioning results.

Our contributions can be listed as follows.

\begin{itemize}
    \item We propose a novel streaming partitioning architecture, which outperforms state-of-the-art solutions in terms of quality and scalability, for big web graph analytics.
        %The time and space complexities of our algorithm, CLUGP, are bounded by $O(3|E|)$ and $O(|V|)$, respectively. %better partition quality and scalability among all stream-based algorithms on web graph. The partition latency of our algotihm is lower than any other low-cut algorithms on large parittion numbers (larger than 16).
    \item We study a new streaming clustering algorithm optimized for vertex-cut partitioning, by extending previous edge-cut streaming clustering algorithms.
    \item We provide a new method for mapping generated clusters to vertex-cut partitions by modeling the process by game theories. We theoretically prove the existence of Nash equilibrium and quality guarantees.
    \item We set up the parallel mechanism for CLUGP, getting rid of
    the computation bottleneck caused by frequent global table accessing by heuristic-based streaming algorithms.
    %that does not need any distributed communication, we also anaylsis why distributed computing nodes can word independently.
    %\item We prove that the time and space complexities of our algorithm, CLUGP, are bounded by $O(3|E|)$ and $O(|V|)$, respectively.
    \item We empirically evaluate CLUGP with real datasets and real distributed graph systems. The results over representative algorithms, such as pagerank and connected component, demonstrate the superiority of our proposals.
    %communication and computation cost for representative graph analytical algorithms, such as pagerank, label propagation and sssp,  on a cluster within $32$ machines integrated with PowerGraph \cite{gonzalez2012powergraph} using CLUSP. The experiments show that CLUSP has lowest communication and computation cost.
\end{itemize}

The rest of the paper is organized as follows.
We first formalize the vertex-cut partitioning problem in Section~\ref{sec:problem}. Then, we propose the CLUGP framework in Section~\ref{sec:framework},
investigate technical details of streaming clustering in Section~\ref{sec:cluster}, and study the partitioning game in Section~\ref{sec:clusp algorithm}. We conduct extensive experiments with real datasets and real systems in Section~\ref{sec:exp}.  We conclude the paper in Section~\ref{sec:con}. Notations of this paper are summarized in Table~\ref{tab:notations}.

\begin{table}
%\small
    \caption{Notations}
    \label{tab:notations}
    \begin{tabular}{ll}
        \hline
        Symbol & Notation \\
        \hline
        $G=(V,E)$ & Directed graph with set of vertices $V$ and edges $E$. \\
        $P$ & The set of $k$ partitions $P = \{p_1,\cdots,p_k\}$. \\
        $P(v)$ & The set of partitions that hold vertex $v$ . \\
        $|p_i|$ & The number of edges within $p_i$. \\
        $G_{S}$ & Edge streaming of the graph $G$. \\
        $G_C$ & The cluster set of graph $G$, $G_C = \{c_1, \cdots, c_m\}$. \\
        $|c_i|$ & The number of intra-cluster edges of $c_i$, $|c_i| = |e(c_i, c_i)|$. \\
        $m$ & The number of clusters, i.e., $|G_C| = m$. \\
        $\varphi(a_i)$ & The individual cost function of $c_i$ under strategy $a_i$.\\
        % $a^*$ & The pure strategy set of a startegic game. \\
        % $a_{-i}$ & The strategies of all players except player $i$. \\
        % $a_i$ & The strategy of player $i$. \\
        $\Phi$ & The potential function of a strategic game. \\
        $\lambda$ & Normalization factor. \\
        $\tau$ & The imbalance factor. \\
        $e(c_i, c_j)$ & The set of edges that across from  cluster $c_i$ to $c_j$. \\
        $e(c_i, V\backslash c_i)$ & The set of edges that across from cluster $c_i$ to other clusters. \\
        \hline
    \end{tabular}
\end{table} 

\section{Preliminaries}
\label{sec:problem}

\subsection{Vertex-Cut Streaming Partitioning}
Given a directed graph $G=(V,E)$, where $V$ is a finite set of vertices, and $E$ is a set of edges.
\begin{definition}[{\bf Edge Streaming Graph Model}]
    The edge streaming graph model $G_S=\{e_1, e_2, \cdots, e_{|E|}\}$ assumes edges of an input graph $G=(V,E)$ arrive sequentially\footnote{Without losing generality, we assume the edge stream of $G$ arrives in the breadth-first (BFS) order, following the setting of \cite{BCSU3, zhu2016gemini,hua2019quasi}, since most real web graphs are formulated and crawled in BFS order.
}, where each edge $e_i=(u, v)$ indicates a directed edge form vertex $u$ to vertex $v$. %with weight $w$.
\end{definition}

In vertex-cut streaming partitioning, partitioning algorithms perform single- or multi-pass over the graph stream and make partitioning decisions for computational load-balancing and communication minimization.

\begin{problem}[\bf Vertex-Cut Streaming Partitioning]
    Given $k$ partitions $\{p_i\}_{1\leq i \leq k}$, the vertex-cut streaming partitioning algorithm assigns each edge $e_i \in G_S$ to a partition $p_i$, such that $\cup_{1 \leq i \leq k} p_i = E$ and $p_i \cap p_j = \emptyset$ ($i \neq j$). Each partition corresponds to a distributed node, each distributed node uses the divided graph edges to perform distributed graph analytic tasks.
\end{problem}

% The vertices are not unique across partitions and are therefore replicated to several partitions. Therefore, the {\it replication factor} is considered as a general metric for measuring the partitioning quality.

% so that the computation workload is balanced and the communication overhead is minimized.

% \begin{definition}[{\bf Edge Streaming Graph Model}]
%     The edge streaming graph model $G_S=\{e_1, e_2, \cdots, e_{|E|}\}$ assumes edges of an input graph $G=(V,E)$ arrive sequentially\footnote{Without losing generality, we assume the edge stream of $G$ arrives in the breadth-first (BFS) order, following the setting of \cite{BCSU3, zhu2016gemini,hua2019quasi}, since most real web graphs are formulated and crawled in BFS order.
% }, where each edge $e_i=(u, v)$ indicates a directed edge form vertex $u$ to vertex $v$. %with weight $w$.
% \end{definition}

{
    % \color{blue}

    \subsection{Partition Quality}
    \label{subsec:partition-quality}

    %We next discuss how to evaluate the partition quality of the vertex-cut partition algorithm. Actually,
    The main goal of partitioning algorithm is to improve the performance of the upper-level distributed graph processing system, like PowerGraph~\cite{xie2014distributed}. Considering the GAS model of the vertex-centric graph processing system, the graph computing messages are aggregated at the vertices and spread along the outgoing edges. After each iteration step, the master vertex gathers the message sent by mirror vertices, and synchronizes it to mirror vertices. Therefore, the number of edges determines the number of messages, and the number of mirror vertices determines the number of synchronizations, within an iteration.

    To accelerate distributed graph processing, one should, 1) balance the computing time of each distributed node (computing cost); 2) reduce the number of synchronizations (communication cost). For the load balance part, we use the relative load balance $\tau \geq \frac{k \cdot \max|p_i|}{|E|}$ to denote the imbalance among partitions, where $|p_i|$ denotes the number of edges in partition $p_i$. $\tau \geq 1$ is a threshold for imbalance. For the synchronizations part, we use the replication factor $\frac{1}{|V|}\sum_{v\in V}|P(v)|$ to denote the proportion of mirror vertices, where $P(v)$ is the set of partitions holding vertex $v$, and $|P(v)|$ refers to the number of partitions holding $v$.

    %multi-objective optimization problem. Following the same partitioning goal used in
    The vertex-cut partitioning can thus be modelled as an optimization problem \cite{gonzalez2012powergraph,petroni2015hdrf}, as follows.
    \begin{equation}
        \label{eq:metric}
        %\footnotesize
        \begin{split}
            & minimize \frac{1}{|V|}\sum_{v\in V}|P(v)| \quad s.t. \frac{k \cdot \max|p_i|}{|E|} \leq \tau \\
        \end{split}
    \end{equation}
    By minimizing the replication factor, the communication cost during graph computation is also minimized. By balancing the workload balance, the computing task of each computing node can be balanced.
}
% The purpose of vertex-cut streaming partitioning is to minimize the replicate factor and the relative load balance. We can follow the same partitioning goal used in \cite{gonzalez2012powergraph,petroni2015hdrf} as:
% \begin{equation}
%     \label{eq:metric}
%     \footnotesize
%     \begin{split}
%         & minimize \frac{1}{|V|}\sum_{v\in V}|P(v)| \quad s.t. \frac{k \cdot \max|p_i|}{|E|} \leq \tau \\
%     \end{split}
% \end{equation}
% where $P(v)$ is the set of partitions holding vertex $v$, $|P(v)|$ refers to the number of partitions holding $v$, and
% %is the number of partitions that hold vertex $v$ ,
% $|p_i|$ denotes the number of edges in partition $p_i$.
% $\tau \geq 1$ is a threshold for imbalance. By minimizing the replication factor, the communication cost during graph computation is also minimized. By balancing the workload balance, the computing task of each computing node can be balanced.
%
%\subsection{Modularity}

%Modularity, a quality metric for graph clusters, is commonly used to measure the strength of division of a network into clusters or the so-called communities. Graph or network with a high modularity has closely connections between the vertices within clusters but sparse connections between vertices of different clusters. So when it comes to community detection or cluster in network, we usually choose modularity as a optimization method. As a consequence, modularity is an significant graph parameter from a practical point of view \cite{prokhorenkova2016modularity}.

\subsection{Power-law Degree Distribution of Web Graphs}
\label{subsec:Web graph}
%In this subsection, we give a brief introduction about the property of web graph.
% {\color{teal}\st{
% We discuss two properties of web graphs, power-law distribution of vertex degrees, and modularity.
% The former determines partitioning strategy, and the latter paves the road to quality partitioning.}}

% {\color{teal} \st{\it Distribution.} }
According to Kumaret et al. \cite{kumar1999trawling, kumar1999extracting} and Kleinberg et al.\cite{kleinberg1999web}, the degree distribution of web graphs follows \textit{power law approximately}.
That is, given a specific degree $x$, the number of vertices follows power-law distribution, $f(x)\propto x^{-\alpha}$, where $\alpha$ is a constant and $\alpha>0$.
The fact that web graphs are featured with power-law distributions are commonly accepted \cite{albert1999diameter,barabasi1999emergence,barabasi2000scale}. Unfortunately, traditional balanced edge-cut partitioning performs poorly on power-law graphs~\cite{abou2006multilevel,lang2004finding}. Percolation theory~\cite{albert2000error} proves that power-law graphs have good vertex-cuts. Therefore, we study the vertex-cut partitioning strategy for web graphs.

\section{Architecture}
\label{sec:framework}

The CLUGP architecture consists of three steps, which process streamed graph edges in three passes, as shown in Figure~\ref{fig:algorithm}.
First, we improve the method of vertex stream clustering proposed by Hollocou et al. \cite{hollocou2017streaming} to produce fine-grained clusters (\emph{streaming clustering step}, Section~\ref{subsec:1pass}).
Second, we investigate game theories to assign clusters to a set of partitions, such that the number of edges across partitions is minimized and the storage of partitions is balanced (\emph{cluster partitioning step}, Section~\ref{subsec:2pass}). Last, we propose a heuristic method to transform cluster partitions into edge partitions (\emph{partitioning transformation step}, Section~\ref{subsec:transform}).

\subsection{First Pass: Streaming Clustering}
\label{subsec:1pass}
The first step is to exploit the connections between clustering and partitioning, so that graph structural information can be leveraged to supervise partitioning, laying the foundation for subsequent steps.

%The main goal of this step is how to use the web graph structure information to guide the subsequent partitioning algorithm and then better satisfy two heuristic rules.

\begin{problem}[\textbf{Streaming Clustering}]
Suppose a streaming graph $G_S=\{e_1, e_2, \cdots, e_{|E|}\}$ and the maximum cluster volume $V_{max}$. The problem is to assign each vertex $v$ to one of the $m$ clusters $\{c_i\}_{1 \leq i \leq m}$, such that the edge-cutting is minimized. Notice that conditions $\cup_{1\leq i \leq m} c_i = V$ and $|c_i| \leq V_{max}$ should be met. The output is a table mapping a vertex to a cluster, i.e., $\{\langle v_i, c_j\rangle\}$.
\end{problem}

The graph clustering can potentially be used for exploiting the structural information of web graphs.
However, there is no existing solution for streaming edge clustering which can be directly used for vertex-cut partitioning.
In our work, we improve the vertex streaming clustering algorithm~\cite{hollocou2017streaming} for adapting to vertex-cut graph partitioning. %, since the problems of vertex and edge clustering have many things in common that can be bridged up.
The challenge is that clustering and partitioning are with different optimization targets.
The goal of vertex clustering is to minimize edge-cutting, while vertex-cut partitioning is on minimizing vertex replicas.
We use an example to show the difference of the two optimization targets in Figure~\ref{fig:algorithm}.

% We find that the modularity-based graph vertex clutstering can be a good guide for vertex-cut partitioning. The advantages of modulrity-based streaming graph clustering are shown as follows. First, as mentioned in Section~\ref{subsec:Web graph}, web graph has the characteristics of high modularity score, the web graph structure information can be detected with high quality by modularity-based graph clustering. Second, during the process of streaming algorithm, the cluster to which vertex belongs can be easily adjusted since the vertex will appear multiple times in the edge stream. Therefore, we can refine the previous inappropiate clustering due to the incomplete information easily. However, the edge will only appear once in the edge stream. So, it is inconvenient to refine the postition of the edge which has already been clustered, unless we use the buffer/slide window to record the previous edge \cite{mayer2018adwise, zhang2017graph, patwary2019window}. Actually, these methods are time consuming and memory overhead, and cannot handle very large-sacle graphs.

%{\color{blue}
%We next show how to bridge the gap between vertex clustering and edge clustering.

\begin{figure*}[h]
    \centering
    \includegraphics[width=0.9\linewidth]{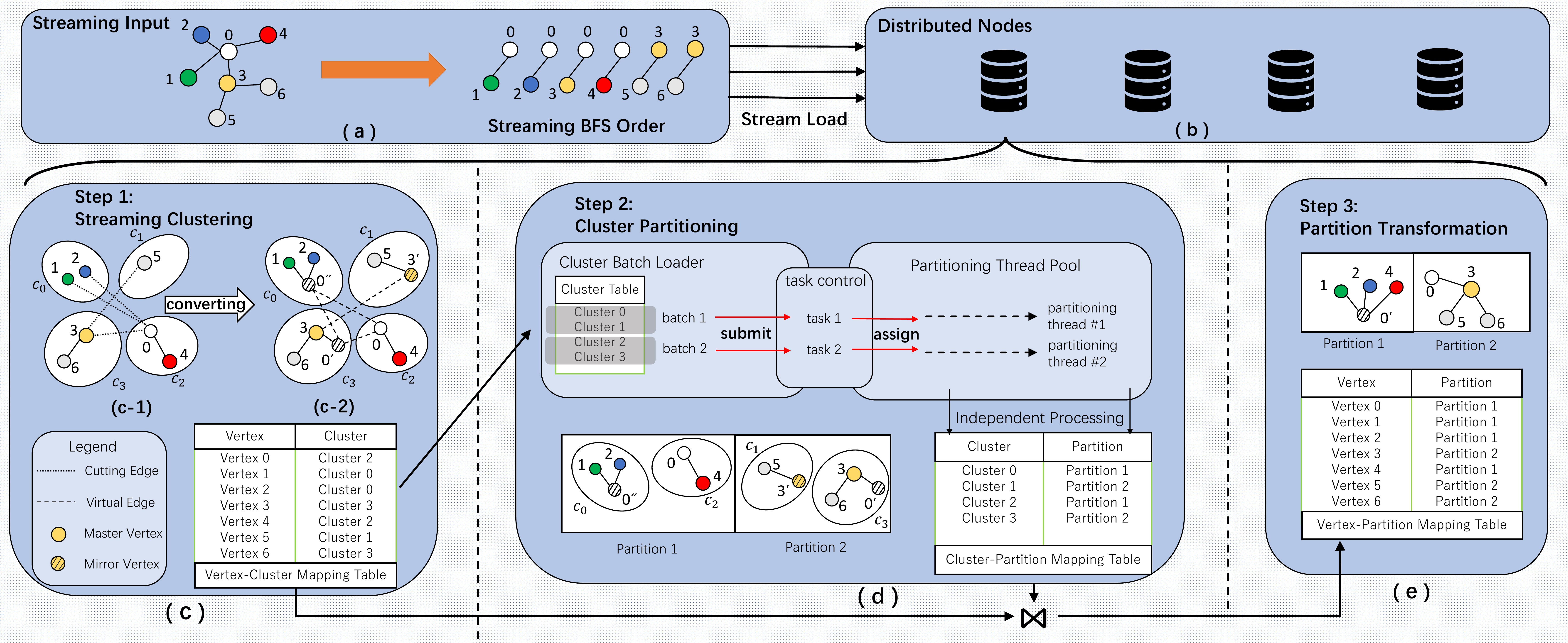}
    \caption{CLUGP Architecture}
    \label{fig:algorithm}
\end{figure*}

For vertex clustering, as shown in Figure~\ref{fig:algorithm} (c-1), vertices $v_0$ to $v_6$ are uniquely assigned to clusters $c_0$ to $c_3$, thus there exist cutting edges, but not vertex replicas. For example, $e(3, 5)$ is a cutting edge generated by the vertex clustering algorithm, while $v_3$ and $v_5$ belong to clusters $c_3$ and $c_1$, respectively. %There are $4$ cutting edges, $e(0,1)$, $e(0,2)$, $e(0,3)$, and $e(3,5)$.

For edge partitioning, as shown in Figure~\ref{fig:algorithm} (c-2), two replicas of vertex $v_0$ and one replica of $v_3$ are generated, as highlighted by dashed circle. The existence of vertex replicas eliminates edge-cutting for ``real'' edges. The dashed lines represent virtual edges, connecting master and mirror vertices. For example, $v_0$ is the master vertex, and $v_0'$ and $v_0''$ are mirror vertices and replicas of $v_0$. Hence, $e(v_0, v_0')$ is a virtual edge\footnote{Without causing any ambiguities, we also call virtual edges as cutting edges for vertex-cut partitioning in the rest of the paper.}.

We therefore propose a new vertex clustering framework, tailored for minimizing vertex replicas.
It produces a coarse-grained vertex-cut partitioning result, in the form of vertex-cluster pairs.
Details are shown in Section~\ref{sec:cluster}.

\subsection{Second Pass: Cluster Partitioning}
\label{subsec:2pass}

The second step is to assign the generated clusters to the given set of partitions, which can be formalized as follows.

\begin{problem}[\textbf{Cluster Partitioning}]
    Suppose $m$ clusters $G_C = \{c_i\}_{1 \leq i \leq m}$ and $k$ partitions $P = \{p_i\}_{1 \leq i \leq k}$.
    %Let $e(c_i, c_i)$ be the set of edges with both end vertices in the cluster $c_i$, and let $e(c_i, V\backslash c_i)$ be the set of replica edges whose source vertex is in the cluster $c_i$ and the other is not. $|c_i| = |e(c_i, c_i)|$ denotes the number of intra-cluster edges of $c_i$
    The problem is to assign each cluster $c_i \in G_C$ to a partition $p_i \in P$, while minimizing edge-cutting and imbalance. The output is a table mapping a cluster to a partition, $\{\langle c_i, p_j\rangle\}$.
\end{problem}

The optimization target of cluster partitioning problem has two parts, load balancing and edge-cutting minimization.
For the load balancing part, without losing generality, we use $\lambda\frac{1}{k}\sum_{p_i \in P}|p_i|^2$ to denote imbalance cost of $k$ partitions \cite{moons2013game,vocking2007selfish}.
Parameter $\lambda$ is for normalization.
It is obvious that the lowest imbalance is achieved, when partitions are of the same size.
For the edge-cutting part, we can use the number of inter-partition (virtual) edges as the cost. By integrating the two costs, we can get the overall cluster partitioning cost function.

\begin{definition}[Cluster Partitioning Cost]
    The overall cluster partitioning cost is defined as:
     \begin{equation}
        \small
        \label{eq:cluster-partition-objfunc}
        Cost =
        \underbrace{\lambda\frac{1}{k}\sum_{p_i \in P}|p_i|^2}_{load ~balancing}
        + \underbrace{\sum_{p_i \in P}|e(p_i, V\backslash p_i)|}_{edge-cutting}
    \end{equation}
    where
    %$\sum_{p_i \in P}|p_i|^2$ is used for load balance and
    $\sum_{p_i \in P}|e(p_i, V\backslash p_i)|$ is the number of cutting edges.
\end{definition}

It can be shown that finding the global optimal solution targeted on Equation~\ref{eq:cluster-partition-objfunc} is NP-hard, by reducing it from the set cover problem. To get a sub-optimal solution, we treat each cluster as a player. Then, the cluster partitioning problem can be modelled as a strategic game.
For a cluster, the selection of a partition can thus be regarded as a rational game strategy, where each cluster affects others’ costs
and meanwhile minimizes its own cost by strategically manipulating its partition choice.
Thus, the optimization problem is transformed into finding the Nash equilibrium of the game, so that each player/cluster minimizes its own cost.
%After the modeling, finding the local optimal solution can be regarded as finding Nash equilibrium that each player gets its minimize individual cost.

However, the retrieval of the Nash equilibrium is compute-bound, a.k.a., the overhead of computation dominates that of I/O.
%the computation time of finding the Nash equilibrium is larger \textbf{than I/O time}, thus will seriously affect the performance of the algorithm.
So, we design a parallel strategy to accelerate the cluster partitioning process.
As shown in Figure~\ref{fig:algorithm}(d), clusters generated are grouped into batches, where each batch is executed by an independent thread to find the Nash equilibrium.
%The feasibility of parallelism is based on the fact that the structurally adjacent clusters are stored next to each other ($c_0$ and $c_1$ in Figure~\ref{fig:algorithm}), since the streaming clustering performed in breadth-first streaming graph order. Furthermore, we can regard that the clusters in each batch are mostly disjoint ($c_0$ and $c_3$ in Figure~\ref{fig:algorithm}), so we can partition these batches in parallel. This way, the time cost of this step can also be bounded by I/O.

%For ease of presentation, we proceed to the discussion of the third step, partition transformation.
More technical details and analysis of cluster partitioning problem are covered by Section~\ref{sec:clusp algorithm}.

\subsection{Third Pass: Partitioning Transformation}
\label{subsec:transform}

%We now present the partition transformation strategy.
%How to deal with the inter-partition edges is the key issue in this step.

By joining the outputs of the first two steps, we can map a vertex to a partition. For mapping an edge to a partition, we utilize partitioning transformation as the third step of CLUGP. It accesses edge streams for further refining the cluster-based partitioning result of the second step.

%To meet the two rules of edge partitioning in Section~\ref{sec:framework}, the third step, partition transformation, fine-tunes the partitioning results of the second step. The details are covered in Algorithm~\ref{alg:transform} in the Appendix.

%cluster partitioning,
%strategy for inter-partition and intra-partition edges are shown in Algorithm~\ref{alg:transform}.

%{
%    \color{blue}

\begin{problem}[\textbf{Partitioning Transformation}]
    Given the mapping table from vertices to partitions, $\{\langle v_i, p_j\rangle\} = \{\langle v_i, c_j\rangle\}\bowtie \{\langle c_i, p_j\rangle\}$, the problem is to transform vertex mapping table $\{\langle v_i, p_j\rangle\}$, to edge mapping table $\{\langle e_i, p_j\rangle\}$, which serves as the partitioning result.
%    a streaming graph $G_S=\{e_1, e_2, \cdots, e_{|E|}\}$,  and the cluster partitioning results of $G_C$.
%    The problem is to assign each edge $e \in G_S$ to a partition $p \in P$ based on the result of cluster partitioning.
%    a set of $k$ partitions $P = \{p_1, \cdots, p_k\}$
\end{problem}

For each edge $e(u, v) \in G_S$, partitions $P(u)$ and $P(v)$ are accessed to determine which partition $e$ is assigned to. The two partitions are retrieved based on the joining results of the first two steps. Notice that we do not explicitly maintain the joining results for reducing memory cost.
Instead, one can quickly map a vertex to a partition by querying the two mapping tables sequentially.
The determination of $e$ assignment has been addressed in previous two steps, following the optimization target of edge-cutting and imbalance. The de facto assignment of edges is implemented in the third step, by traversing the streaming graph. The details are covered in Algorithm~\ref{alg:transform}.

\begin{algorithm}[h]
    \footnotesize
    \caption{Partition Transformation}
    \label{alg:transform}
    \hspace*{\algorithmicindent} \textbf{Input} Cluster Partition Strategy $a^*$, Cluster Set $clu[]$, Vertex \\ \hspace*{\algorithmicindent} Degree $deg[]$, Load Balance Factor $\tau$ \\
    \hspace*{\algorithmicindent} \textbf{Output} Partition Result
    \begin{algorithmic}[1]
        \State Let $a_i$ be the partition choice of cluster $c_i$;
        \State Initialize the array of load, $L_{max}=\tau\frac{|E|}{k}$;
        \For{$e(u, v) \in G_S$}
            \State $c_u, c_v \gets clu[u]$, $clu[v]$;
            \State $p_u, p_v \gets a_u, a_v$;
            \If{$|p_u| \geq L_{max}$ \textbf{or} $|p_v| \geq L_{max}$}
                \If{$|p_u| < L_{max}$}
                    \State assign $e$ to $p_u$;
                    \State \textbf{continue};
                \EndIf
                \If{$|p_v| < L_{max}$}
                    \State assign $e$ to $p_v$;
                    \State \textbf{continue};
                \EndIf
                \For{$p_i \in P$}
                    \State assign $e$ to $p_i$ if $|p_i| < L_{max}$;  \hfill $\rhd$ Load Balance
                \EndFor
            \ElsIf{$p_u$ \textbf{equal} $p_v$}
                \State assign $e$ to $p_u$\;
            \Else
                \If{either $u$ or $v$ has mirror vertices}
                    \State assign $e$ to $p_{v}$ or $p_{u}$;
                \Else
                    \State assign $e$ to $p_{u}$ if $deg[v] > deg[u]$;
                    \State assign $e$ to $p_{v}$ if $deg[u] > deg[v]$;
                    \hfill $\rhd$ Reduce Replicas
                \EndIf
            \EndIf
        \EndFor
    \end{algorithmic}
\end{algorithm}

For each edge $e(u, v)$, if neither of $P(u)$ and $P(v)$ can accommodate $e$, then $e$ will be assigned to an underflow partition, for workload balancing (lines 6-14). When $u$ and $v$ are in the same partition, $e$ will be assigned to the partition (lines 15-16). If $u(v)$ has mirror vertices, which means the vertex $u(v)$ has been replicate during step 1 (Section~\ref{sec:cluster}), $e$ will be assigned to the partitions where $u(v)'s$ mirror vertex belongs to (lines 18-19).
Otherwise, the vertex with a higher degree will be cut (lines 21-22) to reduce vertex replicas, similar to~\cite{petroni2015hdrf,mayer2018adwise,patwary2019window}.

During the transformation, there is a user-specified parameter, i.e., imbalance factor $\tau$, on controlling the partition size. Compared to $\tau$, $V_{max}$ of the first step is merely the upper limit of cluster capacities. The purpose of $\tau$ is to further improve partition balancing from coarse-grained cluster-level to fine-grained partition-level. This way, edges that incur partitioning overflowing are moved to underflow partitions, strictly conforming to the system parameter $\tau$.

For the third step, CLUGP traverses the edge stream to perform partition transformation that merely takes $O(1)$ space cost, since we only need a $k$ elements array to store the partition size.
To perform transformation, the query over vertex-to-partition mapping tables only takes $O(1)$ time for each edge. The total time complexity is $O(|E|)$.

This way, our architecture can be well parallelized.
Of the system, each distributed node accesses partial streaming edges and performs the three steps, clustering, game processing, and transformation, locally.
Further, game processing of a distributed node can be parallelized by multi-threading.
%After the three steps, each node stores its partitioning result to the file system or sends to the corresponds distributed nodes concurrently.
After the three steps, the final graph partitioning result is obtained by combining the partial partitioning results of distributed nodes.
%In this way, at the expense of slight partitioning quality, all distributed nodes can run in parallel without inter-node communication.
%} 

\section{Streaming Clustering}
\label{sec:cluster}

In this section, we investigate a new streaming clustering algorithm. In particular, we propose the allocation-splitting-migration framework in Section~\ref{subsec:cluframework}, and conduct theoretical analysis in Section~\ref{subsec:clusteranalysis}.

\subsection{Allocation-splitting-migration Framework}
\label{subsec:cluframework}

% {\color{teal}\st{
% Modularity-based strategies} \cite{girvan2002community, noack2009multi, shiokawa2013fast,hollocou2017streaming} \st{are widely used for web graph clustering.}}
% and community detection, showing good potentials for being applied for web graph partitioning.
The first and only streaming version of graph clustering algorithm, \emph{Holl}, is proposed by Hollocou et al. \cite{hollocou2017streaming}.

{
    % \color{blue}
    Holl presented an \emph{allocation-migration} framework for streaming clustering. However, Holl cannot be directly applied for graph partitioning, because the {\it allocation-migration} framework of Holl incurs high replication factors. CLUGP improves Holl by adding a splitting operation, and thus construct a new \emph{allocation-splitting-migration}. We will prove that the splitting operation can decrease the replication factor.

\begin{figure}[h]
    \centering
    \subfigure[before splitting]
    {\includegraphics[width=0.25\linewidth]{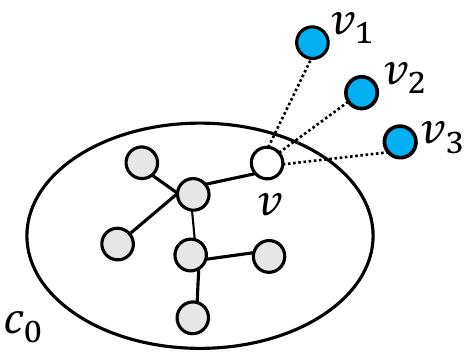}}
    \subfigure[CLUGP with splitting]
    {\includegraphics[width=0.35\linewidth]{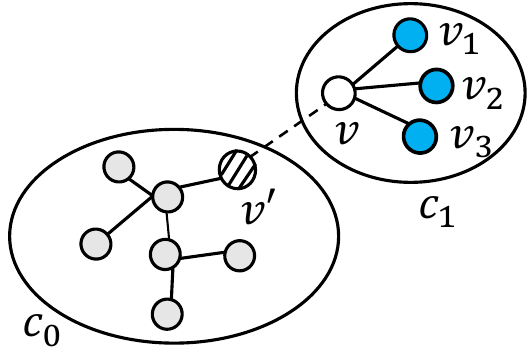}}
    \subfigure[Holl without splitting]
    {\includegraphics[width=0.3\linewidth]{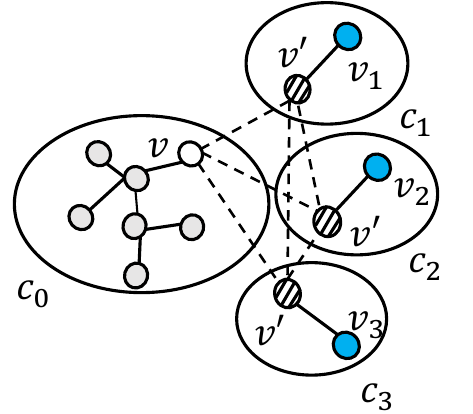}}
    \vspace{-5pt}
    \caption{An Example of Cluster Splitting for Streaming Clustering {\small (Vertex IDs follow BFS order, e.g., $G_S = \{e(v,v_1),e(v,v_2),e(v,v_3),\cdots\}$}}
    \label{fig:cluster-example}
\end{figure}
%\vspace{-5pt}
Consider the example in Figure~\ref{fig:cluster-example}.
Suppose cluster $c_0$ reaches the maximum cluster volume $V_{max}$, in Figure~\ref{fig:cluster-example}(a).
In Holl, to handle incoming edges $e(v,v_1)$, $e(v,v_2)$, and $e(v,v_3)$, cluster $c_0$ remains as it is, while new clusters, $c_1$, $c_2$, and $c_3$, are generated to accommodate successive streaming edges, as shown in Figure~\ref{fig:cluster-example}(c).
According to the {\it allocation-migration} mechanism of Holl, cluster $c_0$ never splits, so that the master vertex of $v$ is always subordinated to $c_0$, and the mirror vertex $v^{\prime}$ exists in $c_1$ to $c_3$. After clustering, the number of master vertices is $10$, and the number of mirror vertices is $3$. Based on Equation~\ref{eq:metric},  the replication factor of Holl is $\frac{10 + 3}{10} = \frac{13}{10}$.

CLUGP adds a splitting operation, as highlighted in Algorithm~\ref{alg:streaming-graph-cluster}. The splitting operation can effectively chop high-degree vertices to reduce replicas in the streaming clustering process, since high-degree vertices tend to form new clusters with subsequent neighboring vertices. In Figure~\ref{fig:cluster-example} (b), with the splitting operation, $c_0$ is split into two clusters, $c_0$ and $c_1$, and the master vertex of $v$ is assigned to new cluster $c_1$ meanwhile generating a mirror vertex $v^{\prime}$ in $c_0$. In this case, the replication factor of CLUGP is $\frac{11}{10}$, which is smaller than that of Holl.
%However, Holl generates $3$ new clusters, $c_1$ to $c_3$, and the replication factor is $\frac{13}{10}$.

%{\bf Say ideas of splitting, e.g., replicate the high-degree vertices first.} The detailed process is shown in Algorithm~\ref{alg:streaming-graph-cluster}.
%For ease of presentation, the improved Holl with splitting operation is called CLUGP.

%{\bf we can add a paragraph to discuss splitting and brief allocation and migration to make the work self-contained. We can add a figure example (for splitting).}

\begin{algorithm}[t]
\footnotesize
    \caption{Streaming Graph Clustering of CLUGP}
    \label{alg:streaming-graph-cluster}
    \hspace*{\algorithmicindent} \textbf{Input} Edge Stream $G_S$, Maximum Cluster Volume $V_{max} = \frac{|E|}{k}$\\
    \hspace*{\algorithmicindent} \textbf{Output} Cluster Set $G_C$, Cluster ID $clu[]$, Vertex Degree $deg[]$

    \begin{algorithmic}[1]
 %       \REQUIRE{Edge Stream $G_S$, Maximum Cluster Volume $V_{max} = \frac{|E|}{k}$}
%        \ENSURE{Cluster Set $G_C$, Cluster ID $clu[]$, Vertex Degree $deg[]$}
        \State Initialize the array of degree, cluster, and volume;
        \State $vol(c_i)$ returns the volume of cluster $c_i$\;
        \For{$e(u,v) \in G_S$}
            \If{$clu[u]$ is NULL \textbf{or} $clu[v]$ is NULL}
                \State Assign a new cluster ID for $u$ or $v$; \hfill %$\rhd$ Allocation\\
            \EndIf
            \State $c_u \gets clu[u]$, $c_v \gets clu[v]$;
            \State $deg[u] \gets deg[u] + 1$, $deg[v] \gets deg[v] + 1$;
            \State $vol(c_u) \gets vol(c_u) + 1$, $vol(c_v) \gets vol(c_v) + 1$;
            \If{$vol(c_u) \geq V_{max}$}
                \State Assign a new cluster ID for $u$;
                \State $c_u^{\prime} \gets clu[u]$, mark $u$ as divided vertex;
                \State $vol(c_u) \gets vol(c_u) - deg[u]$;
                \State $vol(c_u^{\prime}) \gets vol(c_u^{\prime}) + deg[u]$;
            \EndIf
            \If{$vol(c_v) \geq V_{max}$}
                \State Assign a new cluster ID for $v$;
                \State $c_v^{\prime} \gets clu[v]$, mark $v$ as divided vertex;
                \State $vol(c_v) \gets vol(c_v) - deg[v]$;
                \State $vol(c_v^{\prime}) \gets vol(c_v^{\prime}) + deg[u]$; \hfill %$\rhd$ Splitting\\
            \EndIf
            \State $c_u \gets clu[u]$, $c_v \gets clu[v]$;
            \If{$vol(c_u) < V_{max}$ \textbf{and} $vol(c_v) < V_{max}$}
                \If{$vol(c_u) \leq$ $vol(c_v)$}
                    \State Migrate $u$ from $c_u$ to $c_v$;
                    \State Update $vol(c_u)$ and $vol(c_v)$;
                \Else
                    \State Migrate $v$ from $c_v$ to $c_u$;
                    \State Update $vol(c_u)$ and $vol(c_v)$; \hfill %$\rhd$ Migration\\
                \EndIf
            \EndIf
        \EndFor
        \Return{$G_C$}
    \end{algorithmic}
    \AddNote[blue]{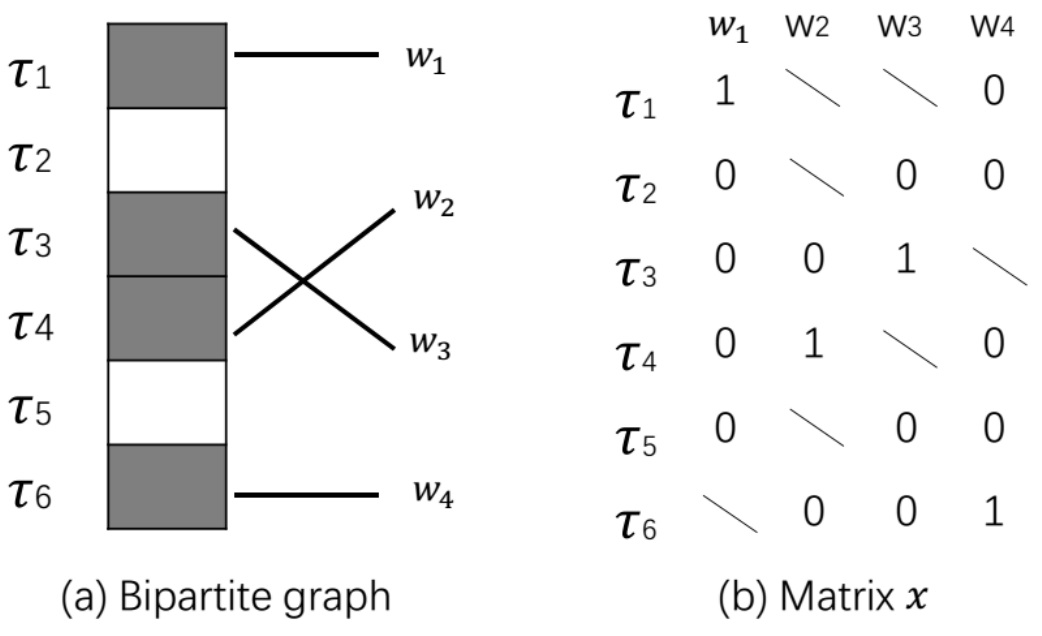}{5}{~~Allocation}
    \AddNote[red]{9}{19}{~~Splitting}
    \AddNote[blue]{20}{26}{~~Migration}
\end{algorithm}

% \vspace{-10pt}

The details about the improved streaming clustering algorithm of CLUGP are covered in Algorithm~\ref{alg:streaming-graph-cluster}.
The clustering process builds clusters in a bottom-up manner, where each cluster initially has one vertex.
For an incoming edge $e(u,v)$ of streaming $G_S$, the two incident vertices of $e$ are assigned to two clusters $c_u$ and $c_v$ (lines 3-5).
%This way, clusters are constructed in a bottom-up manner, where each cluster initially has one vertex.
The volume of a cluster is defined as the sum of the degrees of master vertices in the cluster.
A cluster overflows, if the volume of a cluster exceeds its maximum capacity ($V_{max}$).
Holl handles cluster overflowing, by assigning incoming edges to a new cluster.
CLUGP handles cluster overflowing, by splitting the original cluster into two smaller clusters
%To handle the case of cluster overflow (volume exceeding $V_{max}$), we add the splitting operation
for generating fewer replicas (lines 9-19).
At last, the algorithm migrates an incident vertex of edge $e$ from a smaller cluster to a bigger cluster (lines 20-26). The process repeats until all incoming edges of $G_S$ are processed, so that the cluster set $G_C$ is generated as the output.

%The time complexity of Algorithm~\ref{alg:streaming-graph-cluster} is $O(|E|)$.
In the streaming clustering step, we use the vertex-cluster mapping table to store the cluster that a vertex belongs to. To get the degree of vertices, we also need an array to record the degree. So, the space cost of this step is $O(|V|)$.
The time cost of modifying and querying the mapping table or degree array is $O(1)$.
The process traverses all incoming edges, so the time cost of Algorithm~\ref{alg:streaming-graph-cluster} is $O(|E|)$.

% {
%     \color{blue}
%     (R4.O3) Next, we show how the conversion from vertex clustering to edge partitioning is implemented. A straightforward way is to random assign cut-edges to partitions~\cite{gonzalez2012powergraph} after edge-cut partitioning. Actually, the challenge we face is how to properly apply vertex-cut transformation in the dynamic process of streaming clustering. Furthermore, in the process of streaming clustering, the assignment of each edge will affect the locality structure of the cluster and thus affect the final clustering result. We find that the splitting operation can effectively help improve the clustering result and replication factor.
% }

% Next, we show how the conversion from vertex clustering to edge partitioning is implemented.
During the splitting operation, we mark the vertex that causes cluster splitting as \emph{divided vertex} (lines $11$, $16$).
Then, we can quickly find which vertex has been replicated and which cluster its mirror vertices belongs to.
In Figure~\ref{fig:cluster-example} (b), during the clustering step, vertex $v$ is marked as a divided vertex. So, when processing edge $e(v, v_1)$, we can quickly find that $e(v, v_1)$ should be assigned to $c_1$ and generate a mirror vertex $v^{\prime}$ in $c_0$, thus there exists a cutting edge between $c_0$ and $c_1$ (denoted as dashed lines).
By the way, when both vertices of an edge are marked as divided vertices, we split the vertex with a higher degree vertex and assign the edge to the cluster where lower degree vertex belongs to, which is shown to be effective in reducing replication factor for power-law graphs \cite{petroni2015hdrf,patwary2019window,mayer2018adwise}.

We can get two facts from Algorithm~\ref{alg:streaming-graph-cluster}: a) allocation and splitting operations increase at most one vertex replica at an iteration;
b) migration operation reduces at most one vertex replica at an iteration.
But, a seemingly plausible observation, that the splitting operation triggers more vertex replicas, is not correct, because the splitting operation of CLUGP can reduce total number of replicas.
%We can actually show that the splitting operation incurs less replicas, because it follows the rule of xxxx.
This is guaranteed by Algorithm~\ref{alg:streaming-graph-cluster}.
First, if the splitting operation is not triggered, CLUGP is degenerated into Holl, so that the two have the same replication factor.
Second, if the splitting operation is triggered, CLUGP can derive a smaller replication factor than Holl.
In summary, CLUGP derives a smaller replication factor than Holl.
%Next, we formally prove that the replication factor of CLUGP is smaller than Holl.

\subsection{Analysis}
\label{subsec:clusteranalysis}

We prove that CLUGP can effectively reduce the replicate factor, based on the properties for power-law graphs.
According to \cite{cohen2001breakdown}, for a power-law graph, if we remove $\frac{1}{\theta}$ of vertices with the highest degrees, then the maximum degree $\tilde{M}$ of the remnant subgraph can be approximated by $\tilde{M}=\gamma{\theta}^{1/(1-\alpha)}$, where $\gamma$ is the global minimum degree, and $\alpha$ is the exponent of the power-law graph.
Based on this property, we can get that, given a specific degree $d$, the fraction $\theta$ of vertices satisfying $\{v|degree(v)\geq d\}$, is:
%\vspace{-6pt}
\begin{equation}
\footnotesize
   \theta=\left(\frac{\gamma}{d-1}\right)^{\alpha-1}
   \label{eq:fraction-power-law}
\end{equation}

Equation~\ref{eq:fraction-power-law} can be used for describing the worst case of CLUGP, a.k.a., the highest replication factor of the splitting operation. The details are covered in Theorem~\ref{th:killholl}.

\begin{theorem}
\label{th:killholl}
The upper bound of replication factor of CLUGP is always no larger than that of Holl.
\end{theorem}

\begin{proof}
To prove the theorem, we only need to show that, the upper bound of replicate factor of CLUGP is no larger than the replication factor of Holl, $RF_{clugp} \leq RF_{holl}$.

Given the number of replicas $r$ of vertex $v$, we have $degree(v) \geq d_{min}^{clugp}(r)$, where $d_{min}^{clugp}(r)$ denotes the minimum degree of the vertex $v$, if $v$ has been replicated $r$ times.
Based on Equation~\ref{eq:fraction-power-law}, we can get the maximum number of vertices with $r$ replicas equals to $|V|\left(\frac{\gamma}{d_{min}^{clugp}(r)-1}\right)^{\alpha-1}$, by multiplying $|V|$ with $\theta$.

% \begin{equation}
%     \small
%     %|V^{r}| \leq
%     |V|\left(\frac{\gamma}{d_{min}^{clugp}(r)-1}\right)^{\alpha-1}
% \end{equation}

Considering the worst case, let the number of clusters be $m$, for any vertex $v$ with degree $degree(v)$, there can be most $max(degree(v) - 1, m - 1)$ replicas for $v$. If the $degree(v)$ less than $m$, the worst case can be happened when all $degree(v)$ edges of the vertex are assigned to different clusters. Otherwise, each cluster $\{c_i\}_{1\leq i \leq m}$ has an edge (mirror vertex) of $v$. So, a sequence of maximum replicas can be generated by vertices is $\{m-1, m-2, \cdots, \gamma-1\}$\footnote{For power-law graphs, $\gamma$ is much smaller than $m$, $\gamma << m$. Usually, $\gamma$ equals $1$.}, each replica corresponds to a fraction of vertices $\theta_{r}^{clugp} = \left(\frac{\gamma}{d_{min}^{clugp}(r)-1}\right)^{\alpha-1}$, where $\gamma-1 \leq r \leq m-1$. Thus, we can get the upper bound of replicate factor of CLUGP as follows.
   \begin{equation}
   \label{eq:clugpub}
       \footnotesize
       \begin{split}
           RF_{clugp} \leq & (m-1) \cdot \theta_{m-1}^{clugp}  + (m-2) \cdot (\theta_{m-2}^{clugp} - \theta_{m-1}^{clugp})\\
           & +\dots+ (m-\gamma)(\theta_{\gamma-1}^{clugp}-\theta_\gamma^{clugp})\\
           = & \theta_{m-1}^{clugp}+\dots+\theta_{\gamma}^{clugp}+(m-\gamma)\cdot\theta_{\gamma-1}^{clugp}\\
       \end{split}
   \end{equation}
   Similarly, for Hollocou's algorithm, we can have that:
   \begin{equation}
   \label{eq:hollub}
       \footnotesize
       \begin{split}
           RF_{holl} \leq & \theta_{m-1}^{holl}+\dots+\theta_{\gamma}^{holl}+(m-\gamma)\cdot\theta_{\gamma-1}^{holl}\\
       \end{split}
   \end{equation}
   Based on %Equation~\ref{eq:clugp-holl-degree-cmp} of
   Theorem~\ref{th:low-bound-degree}, we know that $d_{min}^{clugp}(r \geq 2) > d_{min}^{holl}(r \geq 2)$.
   Substituting it into Equation~\ref{eq:fraction-power-law}, we get $\theta_{r}^{clugp} \leq \theta_{r}^{holl}$, then combining it into Equations~\ref{eq:clugpub},~\ref{eq:hollub}, we can get $RF_{clugp} \leq RF_{holl}$.
%    \begin{equation}
%        \small
%        RF_{clugp} \leq RF_{holl}
%    \end{equation}
    So, the theorem is proved.
\end{proof}

\begin{theorem}
    \label{th:low-bound-degree}
    %Let $r$ be the number of replicas times that a vertex is replicated.
    Suppose two vertices $v_c, v_h \in V$, where $v_c$ and $v_h$ are processed by CLUGP and Holl, respectively. If $v_c$ and $v_h$ are both with $r$ replicas, the minimum degree of $v_c$ must be no less than that of $v_h$. Formally, $d_{min}^{clugp}(r) \geq d_{min}^{holl}(r)$.
    %\footnote{The proof can be found in https://github.com/cave-g-f/CLUGP. It will be moved to arXiv for camera-ready version.}.

\end{theorem}

%{
%    \color{blue}
%    \begin{proof}
%        We move the proof to a technical report\footnote{\color{blue}The proof can be found in https://github.com/cave-g-f/CLUGP. It will be moved to arXiv for camera-ready version.}.
%    \end{proof}
%}
 \begin{proof}
     Let $R(v)$ denotes the number of replicas of vertex $v$. For CLUGP, it can be obviously seen from Figure~\ref{fig:cluster-example}(b) that, if $R(v) = 0$, it means the vertex $v$ does not need any replicate, so we have $d_{min}^{clugp}(r) = 1$, if $R(v) = 1$, it means the vertex $v$ has at least one splitting operation, so we have $d_{min}^{clugp}(r) = 2$. Similarly, when $R(v) = 2$, it means we must fill up the cluster $c_1$ and split the vertex $v$ out of $c_1$. To better prove the theorem, we let the degree of $v$'s neighbors equal to the global maximum degree $d_{max}$, which is the worst case of CLUGP, since the splitting operation can be triggered intensively. So on when $R(v) = r \geq 2$ we have the following equation sets:
     \begin{equation}
         \small
         \begin{aligned}
             \left\lbrace
             \begin{array}{lr}
                 1+|Ne_1|+|Ne_1|\cdot d_{max} &= V_{max}\\
                 1+|Ne_1|+|Ne_2|+|Ne_2|\cdot d_{max }&=V_{max}\\
                 1+|Ne_1|+|Ne_2|+|Ne_3|+|Ne_3|\cdot d_{max}&=V_{max}\\
                 &\vdots\\
                 1+|Ne_1|+\dots+|Ne_{r-2}|+|Ne_{r-1}|\cdot d_{max}&=V_{max}\\
                 |Ne_{r}|&=1\\
             \end{array}
             \right.
         \end{aligned}
     \end{equation}
     , where $|Ne_i|$ denotes the number of neighbor vertices that vertex $v$ needed to fill up the cluster $c_i$. And next, we can get the solution as follows:
     \begin{equation}
         \small
         \left\lbrace
         \begin{array}{lr}
             |Ne_i|=\frac{V_{max}-1}{d_{max}}\cdot (\frac{d_{max}}{1+d_{max}})^i&,1\leq i\leq r-1\\
             |Ne_i|=1&,i=r\\
         \end{array}
         \right.
     \end{equation}
     After summing the number of edges needed for $v$, we have:
     \begin{equation}
         \small
         \label{eq:clugp-lower-degree}
         \begin{split}
             d_{min}^{clugp}(r \geq 2)&=1+|Ne_1|+\dots+|Ne_{r-1}|+|Ne_{r}|\\
             &=(V_{max}-1)\left[1-(1-\frac{1}{1+d_{max}})^{r-1} \right]+2\\
         \end{split}
     \end{equation}
     That is, if a vertex $v$ has been replicated $r \geq 2$ times, the degree of $v$ must satisfy $degree(v) \geq d_{min}^{clugp}(r)$. For Holl, since it does not have splitting operation to migrate vertex $v$ out of $c_0$, each neighbour of vertex $v$ will be allocated a independent cluster, thus, we can easily get that $d_{min}^{holl}(r \geq 2) = r - 1$.
     % \begin{equation}
     %     \small
     %     d_{min}^{holl}(r \geq 2) = r - 1;
     % \end{equation}
     Additionally, $d_{min}^{holl}(r) = d_{min}^{clugp}(r)$ when $r \leq 1$, so we only consider the situation that $r \geq 2$. Since for power-law graph $d_{max}\gg 1$, $r-1>0$, thus $\frac{1}{1+d_{max}} \to 0$ and we can get $(1-\frac{1}{1+d_{max}})^{r-1} \sim 1-\frac{r-1}{1+d_{max}}$. Therefore, we can get that:
     \begin{equation}
         \label{eq:clugp-holl-degree-cmp}
         \small
         \begin{split}
             d_{min}^{clugp}(r \geq 2)&=(V_{max}-1) \left[ 1-(1-\frac{1}{1+d_{max}})^{r-1} \right]+2\\
                         &=(V_{max}-1) \cdot (\frac{r-1}{1+d_{max}})+2\\
                         &>(1+d_{max}) \cdot (\frac{r-1}{1+d_{max}})+2 \\
                         &=r+1 > r-1 = d_{min}^{holl}(r \geq 2)\\
         \end{split}
     \end{equation}
     where $V_{max} = \frac{|E|}{k} > d_{max}$.
     Hence, the theorem is proved.
 \end{proof}

\section{Game Theory-based Cluster Partitioning}
\label{sec:clusp algorithm}

%Next, we will study how to find the local optimal solution of the clsuter partitioning problem.
In this section, we study a suboptimal solution for the problem of cluster partitioning.
%We introduce preliminaries about game theories in Section~\ref{subsec:game-theory}.
We formalize the problem of cluster partitioning and prove the existence of Nash equilibrium in Section~\ref{subsec:model}.
We theoretically prove the quality guarantee for the game in Section~\ref{subsec:analysis}.

\subsection{Modeling of Cluster Partitioning Problem}
\label{subsec:model}

In strategic games, a player aims to choose the strategy that minimizes his/her own individual cost. The game continues until a steady state is achieved, in which no player can benefit by unilaterally changing its strategy.

In our work, clusters $\{c_i\}_{i \leq m}$ can be considered as independent and competing players in a strategy game.
For each cluster $c_i$, there can be $k$ choices for choosing a partition.
Let strategy $a_i$ be the partition choice of cluster $c_i$.
Then, the strategy profile $\Lambda=\{a_i\}_{i \leq m}$ consists of the strategies for all clusters.
For a fixed strategy profile $\Lambda$, each strategy $a_i \in \Lambda$ refers to the partition that $c_i$ belongs to.
We use $|a_i|$ to represent the number of edges of the partition that $c_i$ belongs to.
For example, if $a_i$ refers to $p_j$, $|a_i|$ equals to the size of $p_j$, i.e., $|e(p_j, p_j)|$.

Given a strategy profile $\Lambda$, the global deployment cost is denoted as $\varphi(\Lambda)$, and the deployment cost of each cluster $c_i$ is denoted as $\varphi(a_i)$. Intuitively, a lower deployment cost corresponds to a higher partition quality. Based on the cluster partitioning optimization target (Equation~\ref{eq:cluster-partition-objfunc}), the global deployment target $\varphi(\Lambda)$
 can be defined as:
 \begin{equation}
    %\footnotesize
    \label{eq:total cost}
    \varphi(\Lambda) =
    \underbrace{\lambda\frac{1}{k}\sum_{i=1}^{k}|p_i|^2}_{load ~balancing}
    + \underbrace{\sum_{i=1}^{k}|e(p_i, V\backslash p_i)|}_{edge-cutting}
\end{equation}

The game-based solution ensures that the global partitioning optimization target $\varphi(\Lambda)$ (Equation~\ref{eq:total cost}) can be achieved, if each cluster $c_i$'s locally minimized partitioning cost $\varphi(a_i)$ is achieved.
We first explain the local optimization target for each cluster/player. Then, we prove that the local optimization targets of clusters can be integrated as the global target.

The local cost of a cluster has two parts, {\it load balancing} and {\it edge-cutting} (Equation~\ref{eq:individual cost}), which is consistent with the form of the global cost function $\varphi(\Lambda)$ (Equation~\ref{eq:total cost}).

We use variable $|c_i|$ to denote the number of edges of cluster $c_i$, formally, $|c_i| = |e(c_i, c_i)|$.
% The cluster size corresponds to a player's quality in a strategic game.
To ensure the load balance, we should assign the large-scale clusters to the partitions with small size. So, for each cluster $c_i$ and its partition $a_i$, the cost of imbalance can be defined as $\varphi^{load}(a_i) = \frac{1}{k}|c_i||a_i|$. To reduce the number of cut edges, the cluster $c_i$ should be placed in the partition that has the least number of cut edges from other partitions. Therefore, the cost of edge-cut can be defined as $\varphi^{cut}(a_i) = \frac{1}{2}(|e(c_i,V\backslash a_i)| + |e(V\backslash a_i, c_i)|)$. In conclusion, we can get a cluster $c_i$'s cost under the partition $a_i$.
\begin{align}
    \footnotesize
    \label{eq:individual cost}
    \nonumber \varphi(a_i) = & \lambda\varphi^{load}(a_i) + \varphi^{cut}(a_i) \\
    = & \underbrace{\frac{\lambda}{k}|c_i|\cdot|a_i|}_{load~balancing} + \underbrace{\frac{1}{2}(|e(c_i,V\backslash a_i)| + |e(V\backslash a_i, c_i)|)}_{edge-cutting}
\end{align}

%Intuitively, a lower deployment cost corresponds to a higher partition quality. %In the strategic game, a cluster
%
%For ease of presentation, we use a strategy variable $a_i \in P$ to present the partition choice of $c_i$.
%We use variable $|c_i| = |e(c_i, c_i)|$ to represent the quality of cluster $c_i$.

%To reduce the cost on imbalance, each cluster tend to choose the partition with the smallest size. To reduce the cost on the edge-cutting part, the cluster should be assigned to the partition that minimize its edge cut.
%So, we can define the local cost $\phi$ for each cluster as follow.

%\begin{definition}[A Cluster/Player's Cost]
%    For a cluster $c_i \in G_C$ and $k$ partitions, the imbalance cost is defined as $\varphi^{load}(a_i) = \frac{1}{k}|c_i||a_i|$, the edge-cutting cost is defined as $\varphi^{cut}(a_i) = \frac{1}{2}(|e(c_i,V\backslash a_i)| + |e(V\backslash a_i, c_i)|)$, where $|a_i|=\sum_{c_j \in a_i}q_j$ denoted as the number of edges of partition $a_i$.
%    %Based on the load cost and cut edges cost, we can define palyer's individual cost as follow:
%    So, cluster $c_i$'s cost $\varphi_i$ can be represented as follows.
%    \begin{equation}
%        \small
%        \label{eq:individual cost}
%        \varphi(a_i) = \lambda\varphi^{load}(a_i) + \varphi^{cut}(a_i)
%    \end{equation}
%    where $\lambda$ is a normalization factor that balances the value range of two cost metrics.
%\end{definition}

% Here, $\lambda$ is a normalization factor that balances the value range of two cost metrics.
We next show how the local cost function (Equation~\ref{eq:individual cost}) can be derived from the global cost function (Equation~\ref{eq:total cost}).

\begin{equation}
\footnotesize
    \begin{split}
        \varphi(\Lambda)
        % = &\lambda\frac{1}{k}\sum_{i=1}^{k}|p_i|^2 + \sum_{i=1}^{k}|e(p_i, V\backslash p_i)|~~~~(Equation~\ref{eq:total cost}) \\
        = &\lambda\frac{1}{k}\sum_{i=1}^{k}\sum_{c_j \in p_i}|c_j||p_i| + \frac{1}{2}\sum_{i=1}^{k}(|e(p_i, V\backslash p_i)| + |e(V\backslash p_i, p_i)|) \\
        = &\lambda\frac{1}{k}\sum_{i=1}^{m}|c_i||a_i| + \frac{1}{2}\sum_{i=1}^{k}\sum_{c_j \in p_i}(|e(c_j, V\backslash a_j)| + |e(V\backslash a_j, c_j)|) \\
        = &\lambda\frac{1}{k}\sum_{i=1}^{m}|c_i||a_i| + \frac{1}{2}\sum_{i=1}^{m}(|e(c_i, V\backslash a_i)| + |e(V\backslash a_i, c_i)|) \\
        = &\lambda\sum_{i=1}^{m}\varphi^{load}(a_i) + \sum_{i=1}^{m}\varphi^{cut}(a_i) = \sum_{i=1}^{m} \varphi(a_i)
     \end{split}
\end{equation}

Consequently, minimizing the global deployment cost is equivalent to minimizing the set of individual deployment costs.
Then, we can define the Nash equilibrium of the cluster partitioning game as follows.

%Let $\varphi_i \in \varphi$ indicate the cost of the $i$-th player.

%\begin{definition}[Strategic Games]
%    A strategic game can be denoted by $\Gamma=\{N,A,\varphi\}$. The term $N$ indicates the set of palyers and $A$ indecates the strategy space. The strategy of the $i$-th  palyer is denoted by $a_i$ and $a=\{a_1,a_2,\cdots,a_N\}, a\in A$ indicates the strategy set of all players. We define $a_{-i}$ as the strategies of all players except the $i$-th player. The term $\varphi_i \in \varphi$ indicate the cost of the $i$-th player.
%\end{definition}

%Potential game possesses a pure strategy Nash equilibrium, where a global potential function is used to map the costs of all palyers. The game can be considered as an exact potential game if an exact potential fucntion is admitted.

%\begin{definition}[Exact Potential Game]
%    A game is an exact potential game if there exisits an exact potential function $\Phi$ such that $\forall i\in N, a_i,a_i^{\prime}\in A$, there is, $\Phi(a_i^{\prime}, a_{-i}) - \Phi(a_i, a_{-i})=\varphi_i(a_i^{\prime},a_{-i}) - \varphi_i(a_i, a_{-i})$
%\end{definition}

%{\bf connections with overall target of cluster partitioning.}
%The target of a player is to maximize his/her own profit, which is equivalent to minimizing a cluster's deployment cost.

%As mentioned in Section~\ref{subsec:1pass}, {\bf after converting vertex-cut graph partitioning into edge-cut graph partitioning},

\begin{definition}[Nash equilibrium]
    A strategy decision profile $\Lambda^*=\{a_1^*,a_2^*,\cdots,a_m^*\}$ of all clusters is a Nash equilibrium \cite{nash1950equilibrium}, if all clusters achieve their locally optimization targets. This way, no cluster has an incentive for unilaterally deviating the strategy for a lower cost.
\end{definition}

\begin{algorithm}[h]
    \small
    \caption{Nash equilibrium}
    \label{alg:nash equilibrium}
    \hspace*{\algorithmicindent} \textbf{Input} Cluster Set $G_C$, Partition Set $P$, Cluster Neighbors $N[]$ \\
    \hspace*{\algorithmicindent} \textbf{Output} Nash equilibrium
    \begin{algorithmic}[1]
        \State Initial individual cost for each cluster.
        \State Assign each cluster $c_i \in G_C$ to a random partition.
        \Repeat
            \For{$c_i \in G_C$}
                \State $minCost \gets \infty$, $partition \gets \emptyset$;
                \For{$p_i \in P$}
                    \State put cluster $c_i$ into $p_i$;
                    \For{$c^{\prime} \in N[c_i]$}
                        \State update individual cost of $c_i$ based on $c^{\prime}$;
                    \EndFor
                    \State update $minCost$, $partition$;
                \EndFor
            \EndFor
        \Until{Nash equilibrium}
    \end{algorithmic}
\end{algorithm}

Algorithm~\ref{alg:nash equilibrium} shows the process of finding Nash equilibrium.
%In Algorithm~\ref{alg:nash equilibrium},
Initially, clusters of $G_C$ randomly choose partitions with equal probabilities. For each cluster, the partition set $P$ is traversed to get its best strategy (partition choice) that incurs the minimum cost (lines 6-10), and the current strategy is updated if necessary (line 12).
The iterative subroutine (lines 4-13) continues until no cluster updates its partition strategy.

\begin{theorem}
    The time complexity of each round is $\Theta(m)$, the space complexity of game is $O(m)$.
\end{theorem}

\begin{proof}
    It is clear that the number of clusters in $N[]$ is $\sum_{c_i \in G_C}|N(c_i)|$, where $|N(c_i)|$ is the number of neighbors of cluster $c_i$. In Algorithm~\ref{alg:nash equilibrium}, each cluster needs to traverse in total $|N(c_i)|$ clusters to compute the individual cost (lines 8-9). Since all clusters perform this step in a round-robin fashion, the time complexity of each round for all players is $O(\sum_{c_i \in G_C}|N(c_i)|)$. Considering the average case of $\sum_{c_i \in G_C}|N(c_i)|$, the traverse time complexity is $\Theta(m)$. In the game process, we need a table mapping from cluster to partitions, so the space complexity is $O(m)$.
\end{proof}

\subsection{Existence of Nash Equilibrium}

The existence of Nash equilibrium supports the deployment of efficient cluster partitioning, as it enables an alternative solution for retrieving a set of local optimizations instead of a global optimization.
Hence, we proceed to prove the existence of a Nash equilibrium for the cluster partitioning problem.
We start by showing an exact potential function $\Phi(\Lambda)$ (Definition~\ref{eq:epf}), which is defined based on the global cost function $\varphi(\Lambda)$ (Equation~\ref{eq:total cost}).
It is important to note that, if a game has an exact potential function, the game must have Nash equilibrium.

\begin{definition}[Exact Potential Function]%(\textbf{Exact Potential Function})
\label{eq:epf}
    According to Equations~\ref{eq:individual cost} and~\ref{eq:total cost}, the exact potential function of the game can be defined as:
    \begin{equation}
    %\footnotesize
        \label{eq:potential function}
        \Phi(\Lambda) = \lambda\frac{1}{2k}\sum_{i=1}^{k}|p_i|^2 + \frac{1}{2}\sum_{i=1}^{k}|e(p_i,V\backslash p_i)|
    \end{equation}
\end{definition}

The potential function is useful in the analysis of game equilibrium, since the ``incentives'' of all players, e.g., {load balancing and egde-cutting}, are reflected in the function, so that the pure Nash equilibrium can be found by locating the local optima of the potential function. Next, we show the cluster partitioning game is an exact potential game, so that it always converges to a Nash equilibrium.

\begin{theorem}
    \label{th:potential game}
    The cluster partitioning game is an exact potential game.
\end{theorem}

\begin{proof}
    To prove a game is an exact potential game, we need to show that when the partition choice of cluster $c_i$ is changed, the value change of potential function (Equation~\ref{eq:potential function}) is the same as the value change of $c_i$'s individual cost (Equation~\ref{eq:individual cost}), according to Lemma 19.7 of \cite{NisaRougTardVazi07}. %So that the decrease in the value of individual cost is same as the decrease in the value of potential function.

    We define $a_{-i}$ as the strategies of all clusters except cluster $c_i$, under the strategy profile $\Lambda$, formally $a_{-i} = \Lambda - \{a_i\}$. Let $\Phi(a_i^{\prime}, a_{-i})$ be the potential value, when cluster $c_i$ changes its partition choice from $a_i$ to $a_i^{\prime}$, and other clusters remain unchanged. Thus, to prove the theorem, it is equivalent to show $\Phi(a_i^{\prime}, a_{-i}) - \Phi(a_i, a_{-i}) \equiv  \varphi(a_i^{\prime}, a_{-i}) - \varphi(a_{i}, a_{-i})$. According to Equation~\ref{eq:individual cost}, we have $\Delta\varphi = \varphi(a_i^{\prime}, a_{-i}) - \varphi(a_{i}, a_{-i}) = \lambda\frac{1}{k}|c_i|(|a^{\prime}_i| + |c_i| - |a_i|) + \frac{1}{2}(\Delta E)$ and $\Delta E = |e(c_i, a_i)| + |e(a_i, c_i)| - |e(c_i, a^{\prime}_i)|-|e(a^{\prime}_i, c_i)|$.

% \begin{equation}
%     \small
%     \begin{split}
%         & \Delta\varphi = \varphi(a_i^{\prime}, a_{-i}) - \varphi(a_{i}, a_{-i}) = \lambda\frac{1}{k}|c_i|(|a^{\prime}_i| + |c_i| - |a_i|) + \frac{1}{2}(\Delta E) \\
%         & \Delta E = |e(c_i, a_i)| + |e(a_i, c_i)| - |e(c_i, a^{\prime}_i)|-|e(a^{\prime}_i, c_i)|
%     \end{split}
% \end{equation}

    Similarly, according to Equation~\ref{eq:potential function}, we can obtain the potential function difference as: $\Delta\Phi=\Phi(a_i^{\prime}, a_{-i}) - \Phi(a_i, a_{-i}) = \Delta\Phi_{load}+\Delta\Phi_{cut}$. Hence, for the first part we have $\Delta \Phi_{load} = \lambda\frac{1}{2k}(2|a^{\prime}_i||c_i|+{|c_i|}^2-2|a_i||c_i|+|c_i|^2) = \lambda\frac{1}{k}|c_i|(|a^{\prime}_i| + |c_i| - |a_i|)$.
% \begin{equation}
% \small
%     \begin{split}
%         \Delta \Phi_{load} & = \lambda\frac{1}{2k}(2|a^{\prime}_i||c_i|+{|c_i|}^2-2|a_i||c_i|+|c_i|^2) \\
%         & = \lambda\frac{1}{k}|c_i|(|a^{\prime}_i| + |c_i| - |a_i|) \\
%     \end{split}
% \end{equation}
    And the second part can be computed as:
    \begin{equation}
    \footnotesize
        \begin{split}
            \Delta \Phi_{cut} &= \frac{1}{2}[|e(a^{\prime}_i, V\backslash a^{\prime}_i)|-|e(a^{\prime}_i,c_i)|+|e(c_i,V\backslash a^{\prime}_i)| +|e(a_i, V\backslash a_i)| \\
            & \quad + |e(a_i,c_i)| - |e(c_i, V\backslash a_i)|-|e(a^{\prime}_i, V\backslash a^{\prime}_i)| - |e(a_i, V\backslash a_i)|] \\
            & = \frac{1}{2}[|e(c_i,V\backslash a^{\prime}_i)|-|e(a^{\prime}_i,c_i)|+|e(a_i,c_i)| -|e(c_i, V\backslash a_i)|] \\
            & = \frac{1}{2}[|e(c_i, a_i)| + |e(a_i, c_i)| - |e(c_i, a^{\prime}_i)|-|e(a^{\prime}_i, c_i)|] = \frac{1}{2}(\Delta E)
        \end{split}
    \end{equation}
    Therefore, it can be concluded that $\Phi(a_i^{\prime}, a_{-i}) - \Phi(a_i, a_{-i}) \equiv  \varphi(a_i^{\prime}, a_{-i}) - \varphi(a_{i}, a_{-i})$. Thus, Theorem~\ref{th:potential game} is proved.

    % \begin{equation}
    % \small
    % \Phi(a_i^{\prime}, a_{-i}) - \Phi(a_i, a_{-i}) \equiv  \varphi(a_i^{\prime}, a_{-i}) - \varphi(a_{i}, a_{-i})
    % \end{equation}
\end{proof}

Next, we answer $3$ remained questions, a) how to choose the normalization factor $\lambda$; b) how fast the Nash equilibrium can be found; and c) how good is the derived solution.

% Next, we analyze the game process from following aspects: a) how to choose the normalization factor $\lambda$; b) how fast the Nash equilibrium can be found; c) how good is the resulting solution.

\subsection{Game Analysis}
\label{subsec:analysis}

\textbf{Normalization.} We show how to choose the normalization factor $\lambda$. In the game process, the factors of load balancing and edge-cutting have confounded effects over the partitioning optimization. However, the cost of load balancing is counted in millions, and the cost of edge-cutting is counted in thousands, which may not faithfully reflect the weights of the two factors in the total cost. We thus show a strategy for determining $\lambda$.

Without losing generality, we assume the two factors in Equation~\ref{eq:total cost} are of equal importance \cite{armenatzoglou2015real,hua2019quasi}, so that $\lambda\frac{1}{k}\sum_{i=1}^{k}|p_i|^2=\sum_{i=1}^{k}|e(p_i, V\backslash p_i)|$.
%Here, we let the load balance cost be equal to the edge-cut cost in Equation~\ref{eq:total cost}, formally, $\lambda\frac{1}{k}\sum_{i=1}^{k}|p_i|^2=\sum_{i=1}^{k}|e(p_i, V\backslash p_i)|$, therefore, we can have:
Therefore, we can have:
\begin{equation}
%\footnotesize
    \label{eq:beta}
    \begin{split}
        \lambda = \frac{k\sum_{i=1}^{k}|e(p_i, V\backslash p_i)|}{\sum_{i=1}^{k}|p_i|^2}
    \end{split}
\end{equation}

Then, we show the value range of $\lambda$, which is important in analyzing the quality of partitioning.
%may not be comparable. For large graphs, the load balance cost is counted by millions, while the cut edges cost is only counted by thousands. {\bf with the range, how to get the value of $\lambda$}

\begin{theorem}
    \label{th:beta range}
    The value range of $\lambda$ is  $[0, \frac{k^2\sum_{i=1}^{m}|e(c_i,V\backslash c_i)|}{(\sum_{i=1}^{m}|c_i|)^2}]$.
    %and the minimum of $\lambda$ is $0$.
\end{theorem}

\begin{proof}
    %Now, we consider the value range for metric of load balance $\sum_{i=1}^{k}|p_i|^2$ and metric of cut edges $\sum_{i=1}^{k}|e(p_i,V\backslash p_i)|$.
    When all clusters are assigned to the same partition, the load balancing factor reaches the maximum value $(\sum_{i=1}^{m}|c_i|)^2$, while edge-cutting factor gets the minimum value $0$. Conversely, when all clusters are evenly separated to different partitions, the load balancing factor reaches the minimum value $\frac{(\sum_{i=1}^{m}|c_i|)^2}{k}$, while the edge-cutting factor gets the maximum value $\sum_{i=1}^{m}|e(c_i, V\backslash c_i)|$.
    According to the valid variation range of the two factors, we can get $0 \leq \lambda \leq \frac{k^2\sum_{i=1}^{m}|e(c_i,V\backslash c_i)|}{(\sum_{i=1}^{m}|c_i|)^2}$. Therefore, the theorem is proved.
    % \begin{equation}
    % \small
    %     \label{eq:beta range}
    %     0 \leq \lambda \leq \frac{k^2\sum_{i=1}^{m}|e(c_i,V\backslash c_i)|}{(\sum_{i=1}^{m}|c_i|)^2}
    % \end{equation}
    % Therefore, the theorem is proved.
\end{proof}

{
    % \color{blue}

    We next analyze the round complexity and the quality of the game process. The round complexity of a game refers to the number of iterations, i.e., rounds of Algorithm~\ref{alg:nash equilibrium}, taken in finding the Nash equilibrium. A smaller number of rounds corresponds to a faster convergence to the equilibrium, and better efficiency. For the partitioning quality, we use $PoA$ and $PoS$ to measure the quantification of the suboptimality of finding an equilibrium in approaching the optimal solution. $PoA$ $(PoS)$ denotes the ratio of the worst (best) local optimal solution find by Algorithm~\ref{alg:nash equilibrium} to the global optimal solution of Equation~\ref{eq:total cost}. It indicates the upper and lower bound of the Nash equilibrium, reflecting the stability of the game process.
}

\textbf{Round Complexity.} We next study the round complexity of Algorithm~\ref{alg:nash equilibrium}.

\begin{theorem}
    \label{th:round_complexity}
    The number of rounds of cluster partitioning game is bounded by $\sum_{i=1}^{m}|e(c_i, V\backslash c_i)|$.
\end{theorem}

\begin{proof}
    We next study how to prove the Theorem~\ref{th:round_complexity}. Based on the Equations~\ref{eq:potential function}, \ref{eq:beta} and the Theorem~\ref{th:beta range}, we can have $0 \leq \Phi(a)^{load} = \frac{1}{2}\sum_{i=1}^{k}|e(p_i,V\backslash p_i)| \leq \frac{1}{2}\sum_{i=1}^{m}|e(c_i, V\backslash c_i)|$ and $0 \leq \Phi(a)^{cut} \leq \frac{1}{2}\sum_{i=1}^{m}|e(c_i, V\backslash c_i)|$. Hence, we can get the range of potential function $0 \leq \Phi(a) = \Phi(a)^{load} + \Phi(a)^{cut} <= \sum_{i=1}^{m}|e(c_i, V\backslash c_i)|$.

    % In order to estimate the round complexity, we need to scale up the potential function $\Phi(a)$, since the cost reduction in $\Phi(a)$ can be decimal \cite{armenatzoglou2015real,hua2019quasi}.
    Since $ \Phi(a)^{cut} = \frac{1}{2}\sum_{i=1}^{k}|e(p_i,V\backslash p_i)|$ is in the integer domain. It implies that if a cluster changes its current strategy in the game, the reduction of $\Phi(a)$ should be at least $1$. So, the number of rounds will be bounded by $\sum_{i=1}^{m}|e(c_i, V\backslash c_i)|$.
\end{proof}

\textbf{Partitioning Game Quality.} To theoretically bound the partitioning quality, we analyze the worst-case and best-case partitioning quality at the equilibrium, relative to the optimal performance, respectively.
In algorithmic game theories, the two are called {\it Price of Anarchy} (PoA) and {\it Price of Stability} (PoS), which are counterparts to the concept of approximation ratio in algorithm designs.

\begin{definition}{\textbf{(Price of Anarchy)}}
    PoA is the lowest ratio of Nash equilibrium achieved over the optimum value of the overall cost function. Generally, $PoA = \frac{\varphi(\Lambda)_{max}}{\varphi(\Lambda^{opt})}$, where ${\Lambda^{opt}}$ denotes the global optimal strategy minimizing function $\varphi(\Lambda)$.
\end{definition}

\begin{theorem}
    The $PoA$ of the game is bounded by $k + 1$.
\end{theorem}

\begin{proof}
    If the load balancing and edge-cutting factors get their maximum values, simultaneously, the upper bound of $\varphi(\Lambda)$ can be computed as:
    \begin{equation}
    \label{eq:phiub}
    %\footnotesize
        \begin{split}
            \varphi(\Lambda)
            % & = \lambda\frac{1}{k}\sum_{i=1}^{k}|p_i|^2 + \sum_{i=1}^{k}|e(p_i, V\backslash p_i)| \\
             & \leq \frac{k\sum_{i=1}^{m}|e(c_i,V\backslash c_i)|}{(\sum_{i=1}^{m}|c_i|)^2} \cdot \sum_{i=1}^{k}|p_i|^2 + \sum_{i=1}^{k}|e(p_i, V\backslash p_i)| \\
             & \leq (k+1)\sum_{i=1}^{m}|e(c_i,V\backslash c_i)| \\
        \end{split}
    \end{equation}
    %Next, let $a^{opt}$ be the globally optimal strategy that minimizes $\varphi(a)$.
    Similarly, if the two factors get their minimum values, simultaneously, the lower bound of $\varphi(\Lambda^{opt})$ can be computed as:
    \begin{equation}
    \label{eq:optlb}
    %\footnotesize
        \begin{split}
            \varphi(\Lambda^{opt})
            % & = \lambda\frac{1}{k}\sum_{i=1}^{k}|p_i|^2 + \sum_{i=1}^{k}|e(p_i, V\backslash p_i)| \\
            & \geq \frac{k\sum_{i=1}^{m}|e(c_i,V\backslash c_i)|}{(\sum_{i=1}^{m}|c_i|)^2} \cdot \frac{(\sum_{i=1}^{m}|c_i|)^2}{k} \\
            & = \sum_{i=1}^{m}|e(c_i, V\backslash c_i)|
        \end{split}
    \end{equation}
    PoA must be no larger than the quotient of upper bound of $\varphi$ (Equation~\ref{eq:phiub}) and lower bound of $\varphi(\Lambda^{opt})$ (Equation~\ref{eq:optlb}). Thus, we have $PoA = \frac{\varphi(\Lambda)_{max}}{\varphi(\Lambda^{opt})} \leq \frac{(k+1)\sum_{i=1}^{m}|e(c_i, V\backslash c_i)|}{\sum_{i=1}^{m}|e(c_i, V\backslash c_i)|}= k+1$.
    % \begin{equation}
    %     \small
    %     PoA = \frac{\varphi(\Lambda)_{max}}{\varphi(\Lambda^{opt})} \leq \frac{(k+1)\sum_{i=1}^{m}|e(c_i, V\backslash c_i)|}{\sum_{i=1}^{m}|e(c_i, V\backslash c_i)|}= k+1
    % \end{equation}
    Therefore, the theorem is proved.
\end{proof}

\begin{definition}{\textbf{(Price of Stability)}}
    Price of Stability(PoS) is the highest ratio of Nash equilibrium over the optimum value of overall cost function. Generally, $PoS = \frac{\varphi(\Lambda^{\prime})}{\varphi(\Lambda^{opt})}$, where $\Lambda^{\prime}$ is the best Nash equilibrium strategy that minimize the potential function $\Phi(\Lambda)$.
\end{definition}

\begin{theorem}
    The $PoA$ of the game is bounded by $2$.
\end{theorem}

\begin{proof}
    Based on Equations~\ref{eq:individual cost} and \ref{eq:potential function}, for any partition strategy $a$, we can get that $\Phi(\Lambda) \leq \varphi(\Lambda) \leq 2\Phi(\Lambda)$.
    % \begin{equation}
    %     \small
    %     \Phi(\Lambda) \leq \varphi(\Lambda) \leq 2\Phi(\Lambda)
    % \end{equation}
    Since $\Phi(\Lambda^{\prime}) \leq \Phi(\Lambda^{opt})$, we can have $\varphi(\Lambda^{\prime}) \leq 2\Phi(\Lambda^{\prime}) \leq 2\Phi(\Lambda^{opt}) \leq 2\varphi(\Lambda^{opt})$.
    % \begin{equation}
    %     \small
    %     \varphi(\Lambda^{\prime}) \leq 2\Phi(\Lambda^{\prime}) \leq 2\Phi(\Lambda^{opt}) \leq 2\varphi(\Lambda^{opt})
    % \end{equation}
    Thus, we can conclude that $PoS = \frac{\varphi(\Lambda^{\prime})}{\varphi(\Lambda^{opt})} \leq 2$
    % \begin{equation}
    %     \small
    %     PoS = \frac{\varphi(\Lambda^{\prime})}{\varphi(\Lambda^{opt})} \leq 2
    % \end{equation}
    Therefore, the theorem is proved.
\end{proof}

So far, we analyze the feasibility of game-based cluster partitioning. Next, we proceed to discuss how to parallelize the cluster partitioning process.

\subsection{Parallelization}
\label{subsec:parallel}
%During the game process, clusters of a distributed node can be segmented into batches, which can further be parallelized by multi-threading (Figure~\ref{fig:algorithm} xxx).

The parallelization is enabled by clustering%Algorithm~\ref{alg:streaming-graph-cluster}
, which preserves the graph locality, so that two clusters tend to be adjacent in the graph structure, if their cluster IDs are close. For example, as shown in Figure~\ref{fig:cluster-example} (b), if neighbors of vertex $v$ arrive in the BFS order, vertex $v_1$ stands out to form a new cluster $c_1$ with its neighbors, so that $c_1$ and $c_0$ are structurally adjacent.
Based on the observation, we divide clusters within a distributed node into batches according to cluster IDs, for being further parallelized by multi-threading.
We recommend setting the batch size as a constant integer multiple of $k$, for dividing clusters equally into partitions. Otherwise, the solution space of the cluster partitioning problem can be enlarged, because one has to consider more possibilities for balancing, which increases the overhead of finding the Nash equilibrium.

Without parallelization, according to Algorithm~\ref{alg:nash equilibrium}, the time cost for each round of finding Nash equilibrium is $\Theta(m)$, and the round complexity is far less than $|E|$ (Theorem~\ref{th:round_complexity}), so the average time complexity of cluster partitioning step is $\Theta(|E|)$. Additionally, the space cost of this step is $O(m)$. With parallelization, the average time complexity can be approximated as $\Theta(|E_{b}^{avg}| * \frac{m}{batchsize \times threads\_num})$, where $|E_{b}^{avg}|$ is the average number of the intra-cluster edges in each batch. The space cost of this step is $O(bathsize \times threads\_num)$. Each thread holds clusters of the batch size, which is far less than $O(|V|)$.

\section{Experiments}
\label{sec:exp}

%Using real web graphs, we conducted three sets of experiments to evaluate our vertex-cut partitioning algorithm for a) partition quality, b) parallel scalability, and c) validation of streaming cluster and game theory.

%We present the experiment setup in Section~\ref{subsec:setup}. We report detailed experimental results in Section~\ref{subsec:result}.

\subsection{Experiment Setup}
\label{subsec:setup}

{
    % \color{blue}
    {\bfseries Datasets.} We used four real web graphs, UK, Arabic, WebBase, and IT, as listed in Table~\ref{ret:dataset}. Although the scope of the manuscript is on web graphs, we also test the partitioning quality on real social graph Twitter.
}

\begin{table}[h!]
    \centering
    \caption{Details of real-world Web graphs}
    \label{ret:dataset}
    \begin{tabular}	{c|c|c|c|c}
        \hline
        % \toprule
        Alias & Source & $|V|$ & $|E|$ & Size\\
        \hline
        \hline
        UK & uk-2002 \cite{BCSU3} & 19M & 0.3B & 4.7GB \\
        Arabic & arabic-2005 \cite{BCSU3} & 22M & 0.6B & 11GB \\
        WebBase & webbase-2001 \cite{BoVWFI,BRSLLP} & 118M & 1.0B & 17.2GB\\
        IT & it-2004 \cite{BoVWFI,BRSLLP} & 41M & 1.5B & 18.8GB\\
        Twitter & twitter \cite{BoVWFI,BRSLLP} & 41M & 1.4B & 18.3GB\\
        % \bottomrule
        \hline
    \end{tabular}
\end{table}

{
    % \color{blue}

    {\bf Competitors.} We consider $5$ competitors for evaluating the performance of vertex-cut partitioning, as shown in Table~\ref{tab:characteristics}, where HDRF is considered as the state-of-the-art vertex-cut streaming partitioning algorithm\footnote{The source codes of Hashing, DBH, HDRF, Greedy are provided by the first author of HDRF paper (https://github.com/fabiopetroni/VGP). The source code of Mint is obtained by the first author of Mint upon personal request.}. For a fair comparison, we choose default settings and best streaming orders for each of the competitors, a.k.a., random orders for HDRF, Greedy, Hash, and DBH, and BFS orders for Mint and CLUGP. The default parameters of CLUGP are set as follows. The maximum cluster volume $V_{max}$ is set as $\frac{|E|}{k}$ according to the suggestion of~\cite{hollocou2017streaming}, the imbalance factor $\tau = 1.0$, batch size is set as 6400 and the number of partitioning threads are set to 32. For cluster partitioning game, the normalization factor $\lambda$ is set to its maximum value.
}

%{\bfseries Partitioners.} We implemented the parallel version of CLUGP in Java and compared it with the following: 1) Hashing \cite{gonzalez2012powergraph}, a hash-based random edge-cut partitioner; 2) DBH \cite{xie2014distributed}, a degree-based hash partitioner for vertx-cut; 3) Hdrf \cite{petroni2015hdrf}. a state-of-the-art vertex-cut streaming partitioner with provable bound on vertex replication; 4) Greedy \cite{gonzalez2012powergraph}, a heuristic-based vertex-cut partitioner; 5) Mint \cite{hua2019quasi}, a game-theory based fast streaming algorithm.

%To get a fair comparison when evaluating the effective and efficiency, we assume that streaming edges arrive in random order for Hdrf, Greedy, Hash an DBG, and arrive in BFS order for Mint and CLUGP. Since heuristic-based streaming algorithm is hard for parallization, mover over, the stand alone parallel version of Hdrf proposed by author\footnote{https://github.com/fabiopetroni/VGP} perform poor on parallel(slower than single thread), we only use the single thread for Hdrf and Greedy.

%{\bfseries Parameters Settings.}
%We used the following settings throughout all experiments. For Hdrf \cite{petroni2015hdrf}, we set the parameters $\lambda = 1.1$ and $\epsilon = 1$ to get lower replicate factor. For Mint \cite{hua2019quasi}, we set the parameters $\alpha=0.5, \beta=\frac{k^3|V_b|}{2|E_b|^2}$ which are provided by author, where $k$ is the number of partitions, $|E_b|$ and $|V_b|$ are the number of edges and vertices in each batch respectivley.

{
    % \color{blue}
    {\bfseries Metrics.} We use the replication factor and relative load balance, which are commonly accepted to measure the partitioning quality. Details are shown in Section~\ref{subsec:partition-quality}.
}

{\bfseries Environment.} All algorithms are implemented in Java and run a PC with 20 x Intel(R) Xeon(R) CPU $E5$-$2698v4$ @ $2.20$GHz 40 cores and $256$GB main memory. To test the partitioning quality on real distributed environment, we use docker to simulate 32 computing nodes equipped with PowerGraph \cite{gonzalez2012powergraph}, and allocate one CPU for each computing node.

%\subsection{Results}
%\label{subsec:ret}

\subsection{Results}
\label{subsec:result}
%\label{subsec:quality}

\begin{figure*}[!ht]
    \centering
    \subfigure[RF vs. \#. of Partitions (UK-2002)] {\includegraphics[height=1in,width=1.7in]{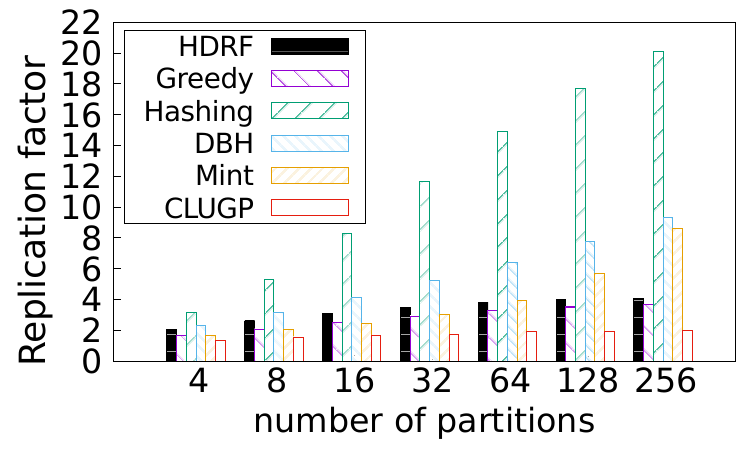}}
    \subfigure[RF vs. \#. of Partitions (Arabic-2005)] {\includegraphics[height=1in,width=1.7in]{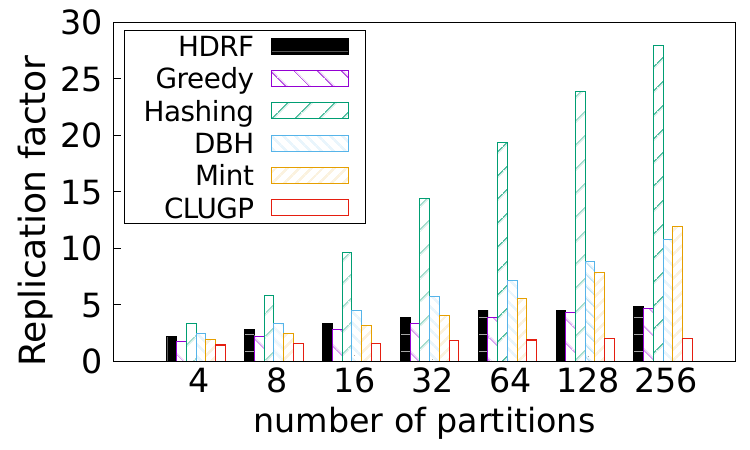}}
    \subfigure[RF vs. \#. of Partitions (WebBase-2001)] {\includegraphics[height=1in,width=1.7in]{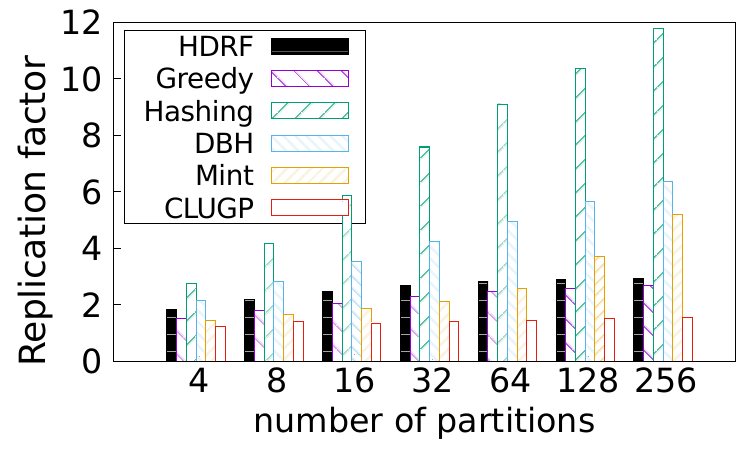}}
    \subfigure[RF vs. \#. of Partitions (IT-2004)] {\includegraphics[height=1in,width=1.7in]{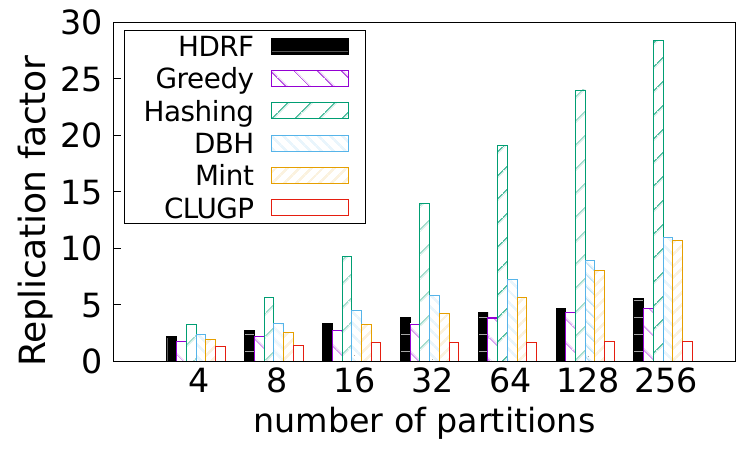}}
    %\vspace{-5pt}
    \caption{Results on Quality (Replication Factor)}
    \label{ret:quality}
%\end{figure*}
%\begin{figure*}[ht]
        \centering
        \begin{minipage}{1\columnwidth}
            \subfigure[RF vs. \#. of Partitions] {\includegraphics[height=1in,width=0.49\linewidth]{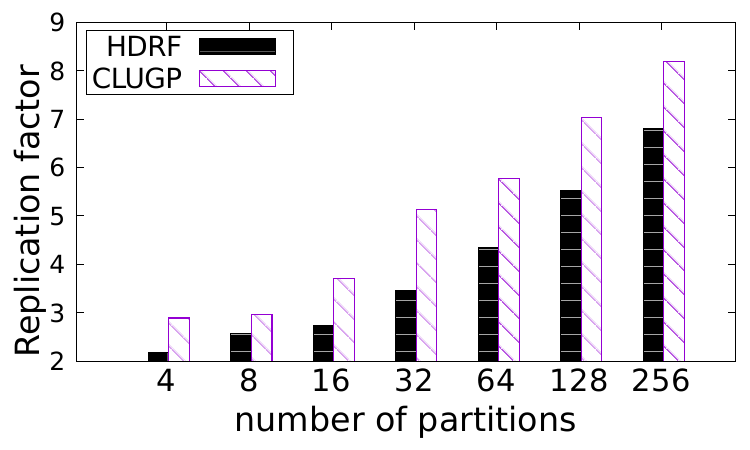}}
            \subfigure[Runtime Cost] {\includegraphics[height=1in,width=0.49\linewidth]{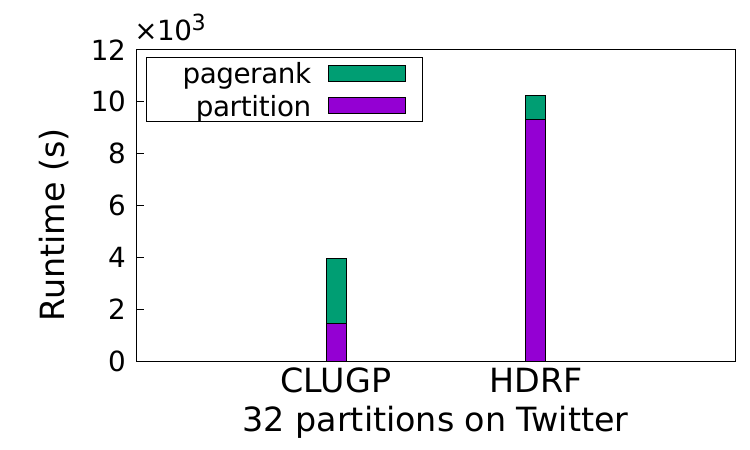}}
            % \subfigure[Space vs. \#. of Partitions] {\includegraphics[height=1in,width=0.32\linewidth]{ret/result/memory/twitter/memory_twitter.csv.pdf}}
            %\vspace{-5pt}
            \caption{Results on Twitter}
            \label{ret:twitter}
        \end{minipage}
        \begin{minipage}{0.49\columnwidth}
            {\includegraphics[width=1\columnwidth]{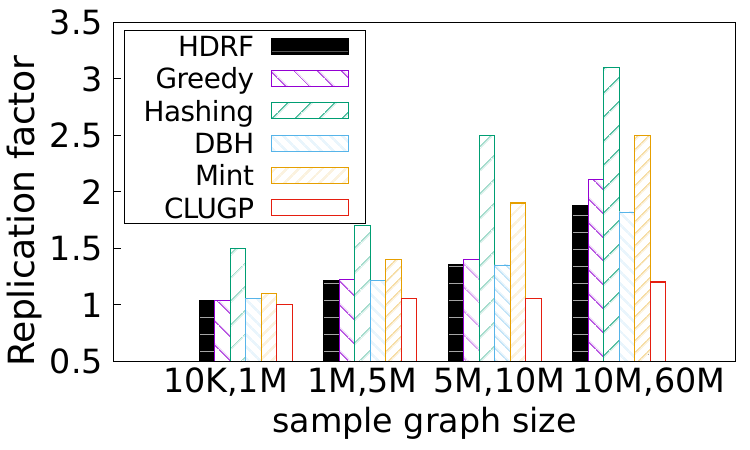}}
            %\vspace{-5pt}
            \caption{Results on Sample Graph}
            \label{ret:sample-graph}
        \end{minipage}
        \centering
        \begin{minipage}{0.49\columnwidth}
            {\includegraphics[width=1\columnwidth]{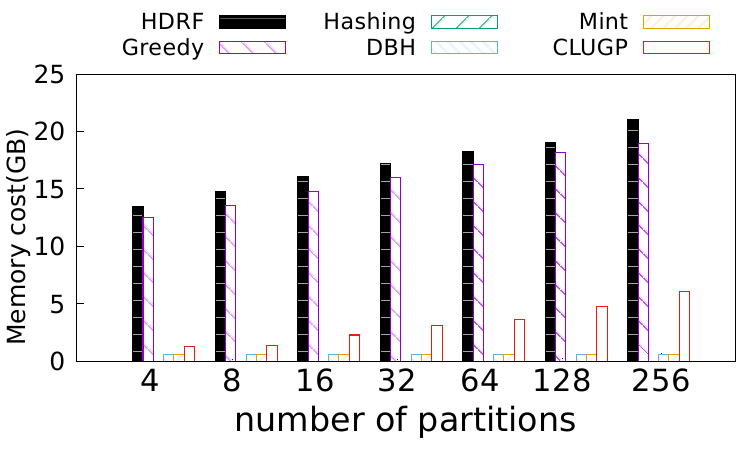}}
            %\vspace{-5pt}
            \caption{Space vs. \#. of Partitions (IT-2004)}
            \label{ret:space}
        \end{minipage}
    \end{figure*}
% \vspace{-5pt}

{\bfseries Replication factor.}
We show the results of quality on $4$ real datasets in Figure~\ref{ret:quality}.
% The replicator factors of all methods increase w.r.t. the number of partitions.
%It shows that heuristic-based methods are better than hashing-based methods, since xxxx.
In all testings, CLUGP outperforms its competitors and the trend of CLUGP is relatively stable.
For example, in Figure~\ref{ret:quality} (b), by increasing the number of partitions from $4$ to $256$, the replication factor of CLUGP increases only about $1.5$ times, while Hashing increases about $10$ times. When the number of partitions equals $256$, the replication factor of CLUGP is only 1/2 of that of HDRF, the best partitioning baseline. More, the second and third best competitors, Greedy and Mint, fall far behind CLUGP in terms of scalability, as it will be shown later.

The lowest replication factors of CLUGP shows the effectiveness of our proposal.
%Some key points are as follows.
The reasons are threefold.
1) For the stream clustering step, the splitting operation helps in reducing the replication factor.
2) For the cluster partitioning step, the cost function is designed to minimize edge-cutting and control the replication factor.
3) For the partition transformation step, the transformation fine-tunes the cluster partitioning result.
%From cluster to vertex, based on the cluster result, partition transformation would then tend to divide high degree vetexes. This directly reduce our goal.
%3)The cost function we design in the cluster partitioning game minimizes the edge-cut, which helps control the replication factor.

{

    % \color{blue}
    We also present the results on social graphs, e.g., Twitter, in Figure~\ref{ret:twitter}. It shows that the replication factor of CLUGP is slightly higher than that of HDRF. But the total task runtime cost, including graph partitioning time and distributed algorithm (e.g, pagerank) execution time, of CLUGP is much lower, because of the partitioning efficiency of CLUGP dominates that of HDRF. We would like to point out that our framework is targeted on web graphs, instead of social graphs.
    % {\color{teal}\st{The modularity of social graphs is often much lower than web graphs, as shown in} Table~\ref{tab:modularity-social}\st{, making the quality of streaming clustering not as good as web graphs.}}

    %The first step of the framework is on clustering which is modularity-based and web graphs often have high modularity values, as shown in Table~\ref{tab:modularity}, showing a good match between data and methodologies. The social graphs, however, do not have high modularity of values, as shown in Table~\ref{tab:modularity-social}.

    We test the performance of CLUGP w.r.t. varied graph sizes. We randomly sample UK-2002 to create a series of graph datasets.
    Figure~\ref{ret:sample-graph} shows that CLUGP has the best partitioning quality. By varying the graph size from $10K$ to $60M$, the replication factor of CLUGP increases only 20\%, while HDRF increases about 80\%.

% \vspace{0pt}
%     \begin{table}[h!]
%         %\small
%         \centering
%         \caption{\color{teal}\st{Modularity of Social Graphs}}
%         \vspace{-5pt}
%         \label{tab:modularity-social}
%         \begin{tabular}{c|c|c}
%             \hline
%             Algorithms & Datasets & Modularity\\
%             \hline
%             \multirow{1}{*}{Improved MOH  \cite{blondel2008fast, panyala2017approximate}} & twitter-2010 &0.47\\
%             \hline
%             \multirow{3}{*}{GPDGP  \cite{onizuka2017graph}} & soc-Pokec & 0.63\\
%             \cline{2-3} & soc-LiveJournal1 & 0.72\\
%             \cline{2-3} & wiki-Talk & 0.56\\
%             \hline
%         \end{tabular}
%     \end{table}
% \vspace{0pt}
}

{
    % \color{blue}
    {\bfseries Load balance.}
    As for the relative load balance, all algorithms achieve $1.0$. We also analyze the influence of relative load balance on replication factor in Figure~\ref{ret:Analysis} (a). It shows the replication factor slightly decreases as the increase of relative load balance.
    In all testings, the quality of CLUGP is stable.
    %Overall, the replication factor of CLUGP is not sensitive to the relative balance. This reality reflects the robustness of the CLUGP.
    % As for the relative load balance, all algorithms achieve the best 1.0, expect for CLUGP, which achieves 1.1 in some large partition numbers. We owe this sightly higher load balance to the iterative algorithms adopted in CLUGP.
    % We argue that 1.1 is adequate in view of the performance gain in Figures~\ref{ret:time}, \ref{ret:pagerank}, and \ref{ret:cc}, because balance factor is a relatively minor factor for iterative graph algorithms \cite{abbas2018streaming}.
}
\begin{figure}[ht]
    \centering
    \begin{minipage}{1.0\columnwidth}
        \subfigure[Time vs. \#. of Partitions (UK-2002)] {\includegraphics[width=0.49\columnwidth]{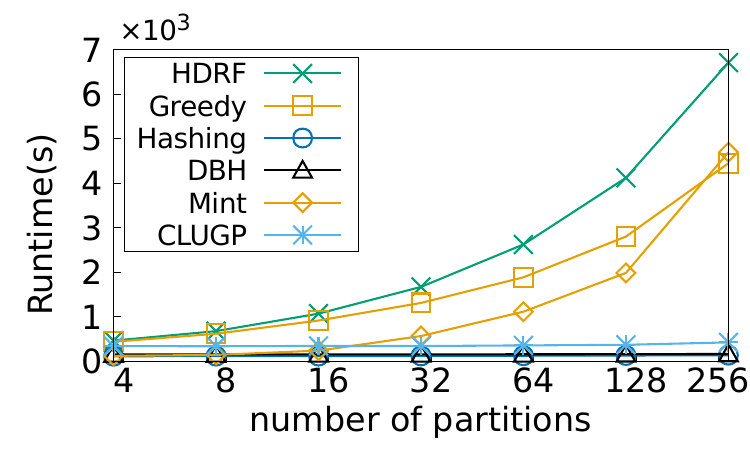}}
        \subfigure[Time vs. \#. of Partitions (IT-2004)] {\includegraphics[width=0.49\columnwidth]{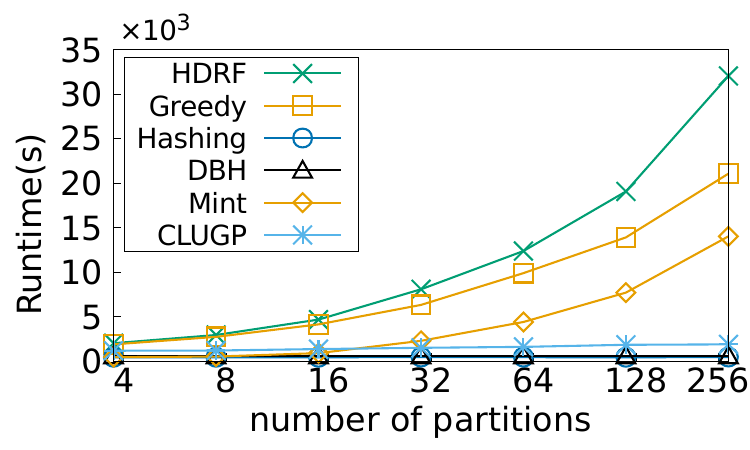}}
        %\vspace{-5pt}
        \caption{Scalability in Terms of Time}
        %\vspace{-5pt}
        \label{ret:time}
    \end{minipage}
    %\vspace{-5pt}
\end{figure}
%\vspace{-5pt}

% \begin{figure}[ht]
%     \centering
%     \begin{minipage}{0.49\columnwidth}
%         {\includegraphics[width=1\columnwidth]{ret/memory_it-2004.csv}}
%         \caption{Space vs. \#. of Partitions (IT-2004)}
%         \label{ret:space}
%     \end{minipage}
%     \begin{minipage}{0.49\columnwidth}
%         {\includegraphics[width=1\columnwidth]{ret/result/rebuttle/sample/sample.csv.pdf}}
%         \caption{Results on Sample Graph}
%         \label{ret:sample-graph}
%     \end{minipage}
% \end{figure}

% \begin{minipage}{0.5\columnwidth}
%     % \subfigure[Space vs. \#. of Partitions (UK-2002)] {\includegraphics[width=0.49\columnwidth]{ret/memory_uk-2002.csv}}
%     \subfigure[Space vs. \#. of Partitions (IT-2004)] {\includegraphics[width=1\columnwidth]{ret/memory_it-2004.csv}}
%     \caption{Scalability in Terms of Space}
%     \label{ret:space}
% \end{minipage}
%     \begin{minipage}{0.5\columnwidth}
%     \subfigure[RF vs. \#. of Partitions] {\includegraphics[width=1\columnwidth]{ret/result/rebuttle/sample/sample.csv.pdf}}
%     \caption{Results on Sample Graph}
%     \label{ret:sample-graph}
% \end{minipage}

\begin{figure*}[ht]
    \centering
    \begin{minipage}{1.5\columnwidth}
        \subfigure[Communication cost vs. Dataset] {\includegraphics[height=1in,width=0.33\linewidth]{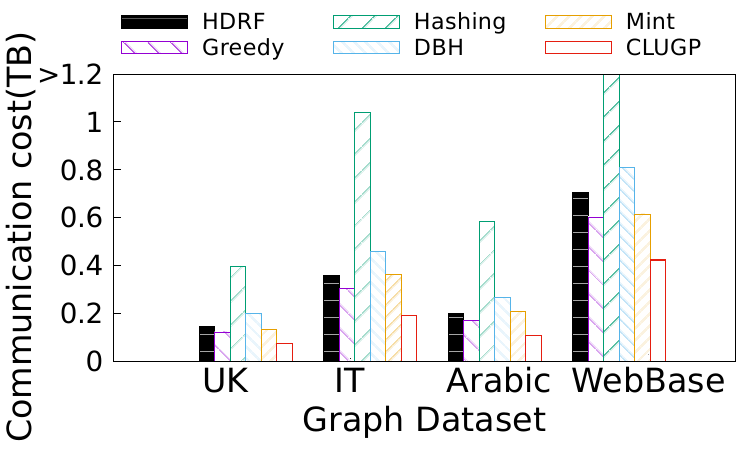}}
        \subfigure[Runtime Cost vs. Dataset] {\includegraphics[height=1in,width=0.33\linewidth]{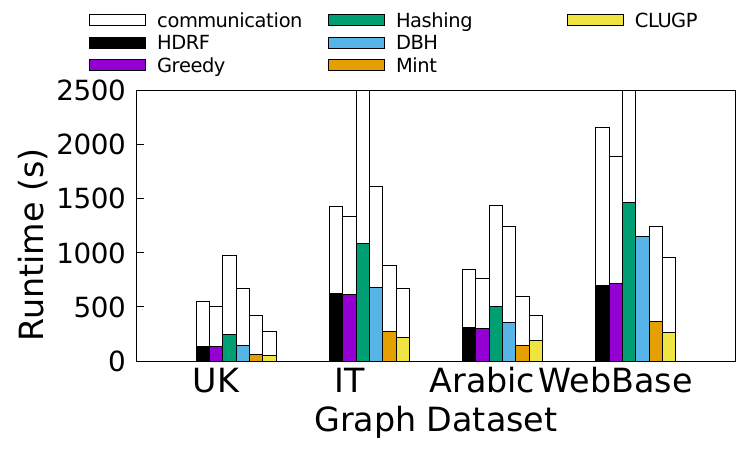}}
        \subfigure[Runtime Cost vs. Latency] {\includegraphics[height=1in,width=0.32\linewidth]{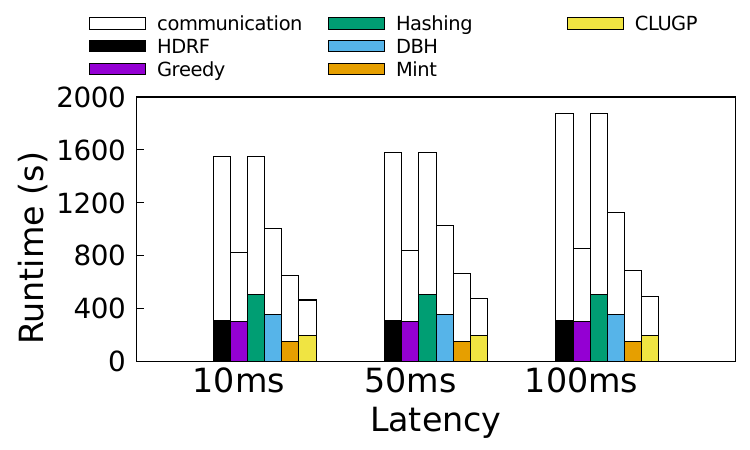}}
        \caption{Results on PageRank (PowerGraph)}
        \label{ret:pagerank}
    \end{minipage}
    \begin{minipage}{0.5\columnwidth}
        \includegraphics[width=1\columnwidth]{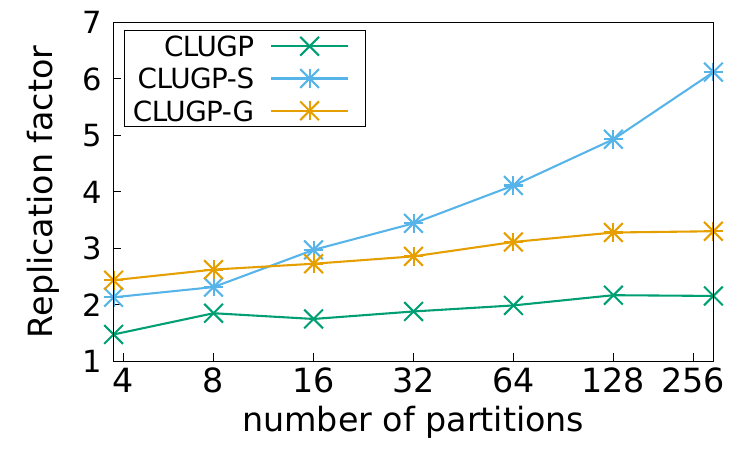}
        \caption{Ablation Study}
        \label{ret:ablation}
    \end{minipage}
%\end{figure*}
%\begin{figure*}
    \centering
    \begin{minipage}{1.0\columnwidth}
        \subfigure[Runtime Cost vs. Algorithms] {\includegraphics[width=0.49\columnwidth]{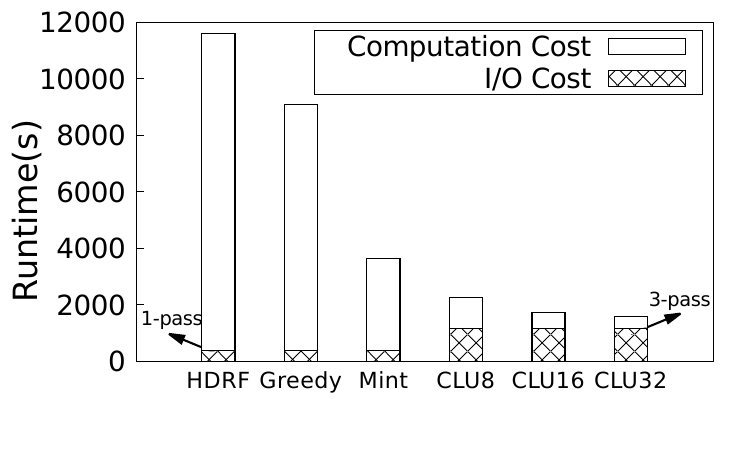}}
        \subfigure[Effect of Batch Size] {\includegraphics[width=0.5\columnwidth]{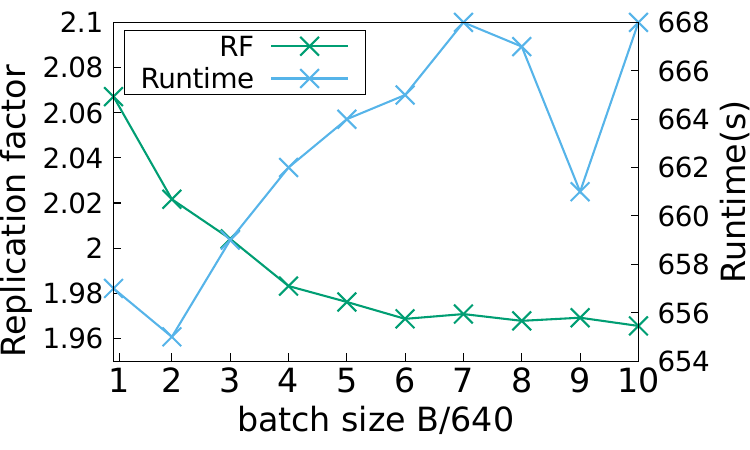}}
        \caption{Parallelization}
        \label{ret:parallel}
    \end{minipage}
    \centering
    \begin{minipage}{1.0\columnwidth}
        \subfigure[Effect of imbalance factor] {\includegraphics[width=0.5\columnwidth]{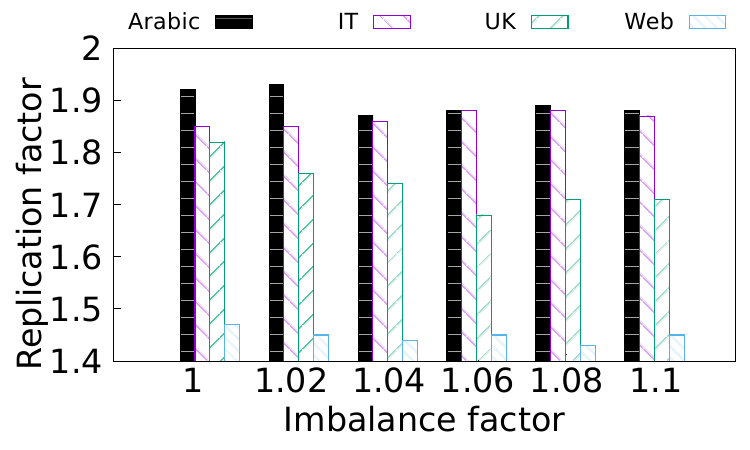}}
        \subfigure[Effect of relative weight] {\includegraphics[width=0.49\columnwidth]{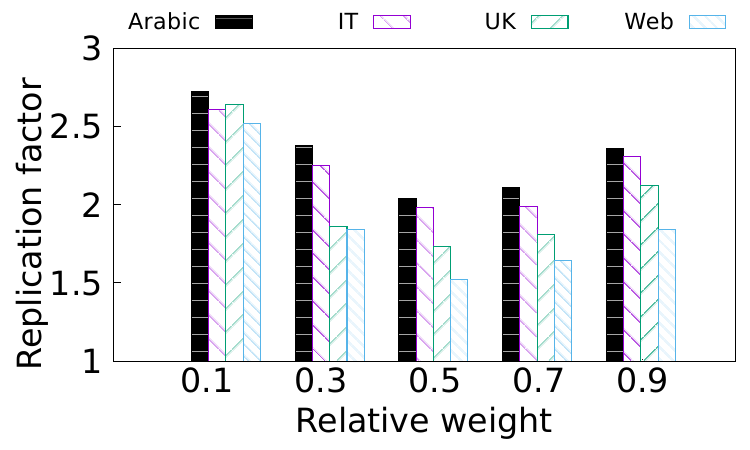}}
        \caption{Analysis}
        \label{ret:Analysis}
    \end{minipage}
\end{figure*}
% \vspace{-5pt}

{\bfseries Space Overhead.}
We measure CLUGP's scalability in terms of space cost against other methods, as shown in Figure~\ref{ret:space}.
% It can be observed that the space overhead increases as the partition number increases, for all competitors.
%All these figures depict the rising tendensy of the memory cost w.r.t. to the number of partitions.
Heuristic-based methods occupy the biggest amount of space, which is about $8$ to $10$ times higher than CLUGP.
Because heuristic-based methods need to maintain the information from all partitions for the optimization purpose.
% Thus, the corresponding space overhead is big.
%keep all the information of each partition before a processing edge, which therefore gains a lot space usage.
By contrast, hashing-based methods take the minimum amount of space. Especially, Hashing takes $0$ space cost, because it merely needs a hash function for making the partition decision. %For it just need a hash function to make every decision.
CLUGP takes larger space cost than Mint,
%is generally about $2$ to $10$ times higher than Mint. The reason is that
since the space complexity of Mint is $O(batch\_size * number\_of\_threads)$, and CLUGP is $O(2|V|)$, as is mentioned in Section~\ref{sec:framework}, Section~\ref{sec:cluster} and Section~\ref{sec:clusp algorithm}. %Though CLUGP take a little more memory, is it worthy of the reduction in replication factor and cost in running iterative algorithm.
Although CLUGP takes bigger space than hashing-based methods and Mint, we argue that the space of several gigabytes are totally affordable for a cluster of hundreds of computing nodes. It is also worthy of the cost for the gaining in replication factor (Figure~\ref{ret:quality}), partitioning efficiency (Figure~\ref{ret:time}), and the computing efficiency (Figure~\ref{ret:pagerank}).

{\bfseries Runtime Scalability.}
We compare the scalability of all methods in terms of runtime cost, in Figure~\ref{ret:time} (a-b).
% It shows that the time costs of all methods increase w.r.t. the number of partitions.
The time costs of heuristic-based methods and Mint increase significantly as the increase of the number of partitions.
In particular, when the number of partitions equals $256$ (Figure~\ref{ret:time} (b)), HDRF takes about $35,000$ seconds for fulfilling the task of graph partitioning.
By contrast, CLUGP and hashing-based methods are not sensitive to the number of partitions. For example, when the number of partitions increase from $4$ to $256$, the runtime cost increases only from $1,162$ to $1,869$ seconds.
% The trends of CLUGP and hashing-based methods are similar.
In all testing, the runtime cost of CLUGP is about $2$ to $3$ times of that of hashing-based methods. We argue that the cost of CLUGP is worthwhile, since the small amount of runtime cost brings in the great benefits of partitioning quality (up to $10\times$ decrease of replication factor in Figure~\ref{ret:quality} (d)).
%We get a better scalability in terms of runtime because CLUGP is an {\bf I/O bounded algorithm, which means that the runtime is almost independent of the number of partitions.} Compared with Mint which adopts edge-partitioning strategy, we achieve a lower runtime by applying cluster partitioning game, indicating that the total game rounds of CLUGP is lower than Mint.

%\subsection
{\bf Performance on Real Systems.}
We examine the performance of partitioning algorithms on real distributed graph systems, PowerGraph. We report the computation and communication cost on pagerank, in Figure~\ref{ret:pagerank}.
{
    % \color{blue} (R4.O4)
    In all testings, CLUGP has the lowest computing time and communication time. The excellent performance is due to the high partitioning quality of CLUGP, including load balancing and low replication factor.
}
In general, hashing-based method perform the worst, and the performance gap is increasing w.r.t. the data volumes.
Heuristic-based methods and Mint are close, but still about 50\% to 100\% higher than CLUGP. In particular, on IT, CLUGP takes about 40\% of communication cost (Figure~\ref{ret:pagerank} (a)), and about 50\% computation cost (Figure~\ref{ret:pagerank} (b)), of the second best method, Greedy.
% Notice that the improvement will be much amplified, if the system is deployed in a real distributed environment rather than a simulated distributed environment, where the communication latency becomes the dominant factor to the overall performance.
{
    % \color{blue}
    To simulate real networking latency, we use PUMBA\footnote{https://github.com/alexei-led/pumba} to vary the RTT from $10$ms to $100$ms. The running time of pagerank under different network latency is shown in Figure~\ref{ret:pagerank} (c). In all testings, CLUGP is the most efficient and the stablest method.
    %In the testing, hashing-based methods is most sensitive to network latency, since hashing-based methods have high replication factor. As the network latency changes, the communication time of CLUGP increases least.
}
% In summary, the results on real systems are consistent with our observation on the quality and scalability testings.

%\subsection{Analysis}

%\begin{figure*}[!ht]
%    \centering
%    \subfigure[Speedup vs. \#. Number of threads] {\includegraphics[height=1in,width=1.7in]{ret/thread_speedup.csv}}
%%    \subfigure[RF vs. \#. Batchsize] {\includegraphics[height=1in,width=1.7in]{ret/batchsize_RF.csv}}
%%    \subfigure[Runing time vs. \#. Batchsize] {\includegraphics[height=1in,width=1.7in]{ret/batchsize_runtime.csv}}
%%    \caption{Scalability in Terms of Threads and Batchsize}
%%    \label{Speedup and Batchsize}
%\end{figure*}

%\begin{figure}[ht]
%    \centering
%    \subfigure[CLUGP vs. CLUGP-S] {\includegraphics[height=1in,width=0.45\linewidth]{ret/cws_it-2004.csv}}
%    \subfigure[CLUGP vs. CLUGP-G] {\includegraphics[height=1in,width=0.45\linewidth]{ret/cwg_it-2004.csv}}
%    \caption{Validation in Terms of Split and Game theory}
%    \label{ret:validation}
%\end{figure}
%In the subsection, we aim to validate the correctness and effectiveness of the Streaming Cluster algorithm and Game Theory of CLUGP.
% \vspace{-11pt}
{\bf Ablation Study.}
    %{\bfseries CLUGP without Split.}
    Splitting operation is the core part of stream clustering. Therefore, we compare the results with and without splitting operation, denoted as CLUGP and CLUGP-S, respectively. The experiment is done on IT by varying the number of partitions from $4$ to $256$. As shown in Figure~\ref{ret:ablation},
    RF of CLUGP is lower than CLUGP-S, in all cases.
    The trend of RF of CLUGP is relatively stable, while the RF of CLUGP-S increases sharply.
    So, we conclude that the splitting operation significantly improves the partitioning quality.
Also, we compare the results with and without game theory-based cluster partitioning, denoted as CLUGP and CLUGP-G, respectively, to show its effectiveness.
%Then, we examine the effectiveness of game theory-based cluster partitioning.
CLUGP-G greedily assigns a big-sized cluster to a small-sized partition. %sorting the volume of clusters and assign the largest cluster into a partition with minimum volume.
%Similarly, we implement CLUGP-G in which we replace the Game Theory part with a basic greedy algorithm and the other remain the same as CLUGP. The greedy algorithms we design will sort the volume of clusters and assign the largest cluster into a partition with minimum volume.
%The experiment settings are the same as CLUGP-S.
The result is depicted in Figure~\ref{ret:ablation}. The replication factor of CLUGP is about 60-70\% lower than CLUGP-S, demonstrating the effectiveness of game theory-based cluster partitioning.
%From the picture, we learn that the RF of CLUGP is smaller than CLUGP-S and the difference of them is around 1 in all cases. Thus, it is ubvious that Game Therory is able to reduce replication factor and improve the quality of partition.

{\bfseries Parallelization.}
%We evaluate the performance of the parallel mechanism in Figure~\ref{ret:parallel} (a).
We evaluate the parallelization performance by varying the number of threads, in Figure~\ref{ret:parallel} (a).
Notice that the total runtime cost consists of I/O cost and computation cost, whereas the latter dominates the total cost. Furthermore, since streaming cluster (step 1) and partitioning transformation (step 3) are all constant time complexity, the cluster partitioning game (step 2) almost occupies all the computation time.
Compared to one-pass streaming partitioning algorithms, e.g., HDRF, Greedy, and Mint (with 32 threads), the total runtime cost of CLUGP is much less. In particular, the runtime cost of CLUGP is about 60\% less than that of the second best competitor, Mint, although the I/O cost of our three-pass streaming partitioning algorithm is three times of that of one-pass competitors.
Also, it shows that, when the number of threads is increased from $8$ to $32$, the computation cost of CLUGP decreases from $1091$ to $429$ seconds, demonstrating good acceleration ratio of our parallelization mechanism.
In particular, the runtime cost of CLUGP (with only $8$ threads) is about 45\% lower than that of Mint (with $32$ threads).
%To figure out speedup of CLUGP, we measure the IO time and partition time of CLUGP versus the number of threads. As can be seen from the picture, when the number of threads range from $8$ to $32$, partition time decreases from $1091$seconds to $429$seconds while IO time always cost about $1150$seconds, which proves CLUGP can speed up with multi-thread implementation and demonstrate the good scalability achieved by our proposed parallel mechanism.
%We also compare the running time of CLUGP with HDRF, Greedy and Mint. Through CLUGP is a three pass algorithm whose IO time is three times as other algorithms, partition time of CLUGP has a huge reduction compared to others. For instance, partition time of CLUGP with $32$ threads is almost one-twentieth  of HDRF and one-third of Mint. On the other hand, this figure also indicates that CLUGP is an IO-bound algorithm, whose running time is mainly restricted by IO process.
Besides, we test the impact of batch size. In Figure~\ref{ret:parallel} (b), it shows that the runtime cost is insensitive to the batch size. With the increasing of batch size, the running time of CLUGP increases slightly.
% In fact, from Figure~\ref{ret:ablation}, we know that the time cost of cluster partitioning(step 2) dominates the total time cost.
When increasing the batch size, although time cost of cluster partitioning game within a batch will be increased, the number of partitioning tasks decreased.

{
    % \color{blue}
    {\bf Relative Weight.} Similar to previous works, we treat the two partitioning metrics in Section~\ref{subsec:partition-quality} as equally important and set the relative weight of Equation~\ref{eq:cluster-partition-objfunc} to $0.5$. We next study the influence of relative weight on partitioning quality, the result is shown in Figure~\ref{ret:Analysis} (b). We can have two observations: 1) in all testings, the replication factor of CLUGP is lower than its competitors; %This is consistent with the ablation study in Figure~\ref{ret:ablation} (a), which shows that the partitioning quality after the first step is already high;
    2) the curve of replication factor of CLUGP is U-shaped with a wide and smooth valley. The replication factor is high for two extremes. When the relative weight equals $0.1$, the optimization target is mostly on the replication factor, so that clusters are mostly put to very few partitions, which is almost equivalent to skipping the game process, resulting in a high replication factor. When the relative weight equals $0.9$, the optimization target is mostly on the load balance, so clusters tend to be sent evenly to the set of partitions. For other valued weights (i.e., the relative weight is in $[0.3, 0.7]$), the variation of replication factor is mostly within $10\%$. We can conclude that the relative weights do not have significant effect on partitioning quality, except extreme cases.
} 

\section{Related Work}
\label{sec:related}
{
    % \color{blue}
    There exists many algorithms for edge-cut and vertex-cut partitioning. Edge-cut partitioning aims to assign vertices into different partitions, while minimizing edge-cutting. METIS~\cite{karypis1996parallel} is an offline algorithm that adopts multi-level heuristics achieving high partitioning quality for edge-cut partitioning. However, efficiency of offline partitioning is low.
%However, the offline algorithms require full knowledge of the graph, thus impractical for large-scale graph.
    Streaming partitioning is considered to be practical for large-scale graph processing~\cite{stanton2012streaming, tsourakakis2014fennel,gonzalez2012powergraph, petroni2015hdrf, xie2014distributed, hua2019quasi}.    LDG~\cite{stanton2012streaming} tends to assigning neighboring vertices into the same partition. FENNEL~\cite{tsourakakis2014fennel} is an edge-cut partitioning algorithm which places a new vertex to the partition holding the most neighboring vertices or holding the least non-neighboring vertices.
    The vertex-cut streaming partitioning is first proposed in~\cite{gonzalez2012powergraph}, and has been proved to be effective on power-law graphs.
    Greedy~\cite{gonzalez2012powergraph} is a heuristic-based partitioning strategy which aims to minimize vertex-cuts.
    HDRF~\cite{petroni2015hdrf} makes use of the skewed distribution of degrees, and cuts the high-degree vertices first to reduce replicas. Similarly, DBH~\cite{xie2014distributed} prioritizes the cutting of high-degree vertices, with hashing-based methods.
    Both Greedy and HDRF need to record the previous results, which are hard to be parallelized. Mint~\cite{hua2019quasi} is a parallel algorithm that achieves a good trade-off between scalability and partitioning quality. Different from previous works, we explore graph clustering for enhancing the partitioning quality, employ streaming techniques for improving the efficiency and break the ties of global structures for boosting system performance.
    %To further reduce the communication cost for low-degree vertices. PowerLyra~\cite{chen2019powerlyra} and IOGP~\cite{dai2017iogp} combine vertex-cut and edge-cut by cutting only high-degree vertices.

} 

\section{Conclusion}
\label{sec:con}

In this paper, we study the problem of edge partitioning for web graphs by proposing a novel restreaming architecture, called CLUGP.
%Particularly, we explore the properties of web graphs for improving the partitioning quality and scalability.
Of the architecture, our techniques can be pipelined as three steps, streaming clustering, cluster partitioning, and transformation.
Compared with state-of-the-art algorithms, CLUGP achieves the best partitioning quality.
Also, we investigate parallelization mechanism to enhance the partitioning scalability.
% Also, we investigate parallelization mechanism to enhance the partitioning efficiency and scalability.
The results on real datasets and distributed graph systems show that the scalability of CLUGP is significantly better than that of one-pass streaming partitioning methods.
% It shows that the runtime cost of CLUGP is significantly lower than that of one-pass streaming partitioning methods, especially when the number of partitions is big.

% \input{app}
\newpage

\bibliographystyle{IEEEtran}
\bibliography{sample}

\end{document}